\documentclass[aps,pra,twocolumn,groupedaddress]{revtex4-1}
\usepackage{graphicx}
\usepackage{epstopdf}
\usepackage{braket}
\usepackage{lipsum}

\usepackage{amsmath,amssymb,amsfonts,amsthm}
\usepackage{multirow}
\usepackage{verbatim}
\usepackage{url}
\usepackage{comment}
\usepackage{mathtools}
\usepackage{epsf,pgf,graphicx}
\usepackage{color}
\usepackage{tikz,pgf}
\usepackage{xcolor}
\usepackage{makecell}
\newtheorem{proposition}{Proposition}
\newtheorem{lemma}{Lemma}
\newtheorem{theorem}{Theorem}

\usepackage{adjustbox}
\usepackage[colorlinks = true,
            linkcolor = blue,
            urlcolor  = blue,
            citecolor = blue,
            anchorcolor = blue]{hyperref}
\usepackage{subcaption}

\begin{document}
\title{Optimal Toffoli-Depth Multi-Controlled Toffoli Decomposition in 2D Qubit Layout}
\author{Anik Basu Bhaumik, Suman Dutta, Anupam Chattopadhyay}
\affiliation{College of Computing \& Data Science, Nanyang Technological University, Singapore}
\email{anikbasu001@e.ntu.edu.sg\\ sumand.iiserb@gmail.com\\ anupam@ntu.edu.sg}
\begin{abstract}
The multi-controlled Toffoli (MCT) gate is a key primitive in quantum arithmetic, oracle construction, and quantum cryptanalysis. Although recent work has established optimal Toffoli-depth MCT decompositions under all-to-all qubit connectivity, their realization on near-term quantum hardware with restricted qubit connectivity remains largely unexplored. While general-purpose quantum mappers can route arbitrary circuits, they do not explicitly exploit the repeated interaction patterns inherent in MCT decompositions.

In our present paper, we study architecture-aware mappings of optimal Toffoli-depth MCT decompositions onto restricted two-dimensional qubit layouts. We begin with a structured geometric placements that preserve the parallelism of state-of-the-art Toffoli and MCT decompositions with no additional depth overhead. We further introduce a motif-based packing framework in which decomposition layers are represented by interaction motifs derived from basic Toffoli gates. By embedding these motifs vertex-disjointly into hardware graphs, we characterize the minimum-size topologies supporting the required qubit resources and derive explicit bounds on the resulting depth overhead under tight qubit budgets. Finally, we compare these bounds with routing-aware placement heuristics and empirically evaluate the effectiveness of embedding different motifs across a range of hardware topologies.
\end{abstract}
\maketitle

\section{Introduction}
\label{sec:intro}
Quantum gates are the fundamental building blocks of quantum circuits~\cite{gates1995}. Unlike classical gates, quantum gates are inherently reversible. Thus, an $n$-qubit quantum gate is mathematically represented by a $2^n\times 2^n$ unitary matrices. The doubly-controlled X-gate, commonly known as the Toffoli gate, is among the most significant quantum gates, with applications in arithmetic circuit design, reversible algorithms, oracle constructions, and error-correction subroutines. Consequently, their efficient synthesis is critical for realizing scalable quantum computation on fault-tolerant hardware. Within the Clifford+T universal gate set, the cost of implementing such gates is typically dominated by non-Clifford T gates, making T-count and T-depth key metrics for assessing circuit feasibility.

Over the last few decades, numerous efforts have been made to reduce the resource requirements for the efficient implementation of multi-controlled Toffoli (MCT) gates. For basic Toffoli gates, state-of-the-art techniques employ measurement-based uncomputation (see FIG.~\ref{fig:toffoli}), achieving Toffoli decompositions using only four T gates~\cite{gidney2018quantum, jaques2020eurocrypt}. On the other hand, recent works on MCT decomposition use the conditionally clean ancilla technique~\cite{nie2024arxiv, khattar2025quantum, dutta2025pra} to reduce the ancilla count and circuit depth.
\begin{figure}[htbp]
	\centering
	\begin{subfigure}{.49\textwidth}
		\centering
		\includegraphics[width=1\linewidth]{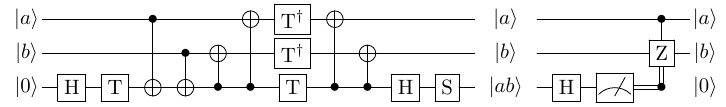}
		\caption{}
		\label{fig:logicalAND}
	\end{subfigure}
	\begin{subfigure}{.49\textwidth}
		\centering
		\includegraphics[width=1\linewidth]{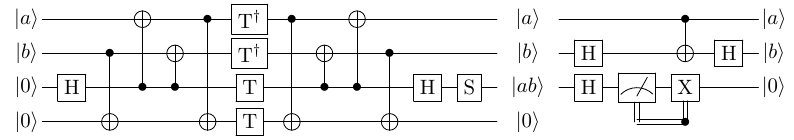}
		\caption{}
		\label{fig:Jaques}
	\end{subfigure}
	\caption{Logical circuits for basic Toffoli decompositions with measurement-based uncomputation: (a) without ancilla \cite{gidney2018quantum}, and (b) with a single reusable ancilla \cite{jaques2020eurocrypt}.}
	\label{fig:toffoli}
\end{figure}

Although these recent advancements have significantly improved logical-level performance, their practicality remains constrained by the architectural realities of modern quantum processors. Contemporary hardware platforms, such as superconducting qubits~\cite{nakamura1999nature}, trapped-ion systems~\cite{cirac1995prl}, quantum dots~\cite{quantum-dot1998}, and neutral-atom~\cite{jaksch2000prl,henriet2020quantum}, predominantly employ two-dimensional qubit connectivity. Common examples include rectangular and square-grid lattices~\cite{helmer2009epl}, heavy-hex~\cite{chamberland2020prx}, brickwork layouts \cite{datta2022ismvl}, octagonal \cite{li2023prr} or triangular lattices \cite{yang2025arxiv}, H-tree networks~\cite{leiserson1980htree}, etc.
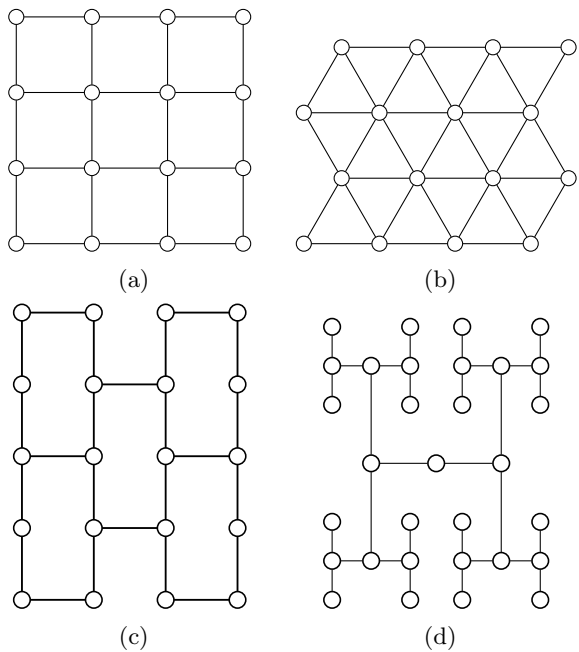
\begin{figure}[htbp]
	\centering
	\begin{subfigure}[b]{.22\textwidth}
		\centering
		\begin{tikzpicture}[
  node/.style={circle, draw, inner sep=2pt},
  ed/.style={black, line cap=round},
  scale=1
]
  \def\W{4} 
  \def\H{4} 

  \foreach \y in {0,...,\numexpr\H-1\relax} {
    \foreach \x in {0,...,\numexpr\W-1\relax} {
      \node[node] (n-\x-\y) at (\x,\y) {};
    }
  }

  \foreach \y in {0,...,\numexpr\H-1\relax} {
    \foreach \x in {0,...,\numexpr\W-2\relax} {
      \draw[ed] (n-\x-\y) -- (n-\the\numexpr\x+1\relax-\y);
    }
  }

  \foreach \x in {0,...,\numexpr\W-1\relax} {
    \foreach \y in {0,...,\numexpr\H-2\relax} {
      \draw[ed] (n-\x-\y) -- (n-\x-\the\numexpr\y+1\relax);
    }
  }
\end{tikzpicture}
		\caption{}
		\label{fig:square}
	\end{subfigure}
	\begin{subfigure}[b]{.22\textwidth}
		\centering
		\begin{tikzpicture}[
  node/.style={circle, draw, inner sep=2pt},
  ed/.style={black, line cap=round},
  scale=1
]
  \def\W{4} 
  \def\H{4} 
  \def\a{1} 
  \def\h{0.8660254} 

  \foreach \y in {0,...,\numexpr\H-1\relax} {
    \foreach \x in {0,...,\numexpr\W-1\relax} {
      \pgfmathsetmacro{\xx}{(\x + 0.5*mod(\y,2))*\a}
      \pgfmathsetmacro{\yy}{\y*\h*\a}
      \node[node] (n-\x-\y) at (\xx,\yy) {};
    }
  }

  \foreach \y in {0,...,\numexpr\H-1\relax} {
    \foreach \x in {0,...,\numexpr\W-1\relax} {

      \ifnum\x<\numexpr\W-1\relax
        \draw[ed] (n-\x-\y) -- (n-\the\numexpr\x+1\relax-\y);
      \fi

      \ifnum\y<\numexpr\H-1\relax
        \ifodd\y
          \draw[ed] (n-\x-\y) -- (n-\x-\the\numexpr\y+1\relax);
          \ifnum\x<\numexpr\W-1\relax
            \draw[ed] (n-\x-\y) -- (n-\the\numexpr\x+1\relax-\the\numexpr\y+1\relax);
          \fi
        \else
          \draw[ed] (n-\x-\y) -- (n-\x-\the\numexpr\y+1\relax);
          \ifnum\x>0\relax
            \draw[ed] (n-\x-\y) -- (n-\the\numexpr\x-1\relax-\the\numexpr\y+1\relax);
          \fi
        \fi
      \fi
    }
  }
\end{tikzpicture}
		\caption{}
		\label{fig:triangular}
	\end{subfigure}\\
	\begin{subfigure}[b]{.22\textwidth}
		\centering
		\begin{tikzpicture}[
  scale=0.95,
  dot/.style={circle, draw=black, line width=0.6pt, fill=white, inner sep=2.2pt},
  edge/.style={line width=0.7pt, black},
]

\def\C{3}   
\def\R{2}   
\pgfmathtruncatemacro{\Ymax}{2*\R}

\foreach \x in {0,...,\numexpr\C-1\relax} {
  \foreach \t in {0,...,\numexpr\R-1\relax} {
    \pgfmathtruncatemacro{\y0}{2*\t + mod(\x,2)}
    \pgfmathtruncatemacro{\cc}{mod(\x+\t,3)}  

  }
}

\foreach \x in {0,...,\C} {
  \foreach \y in {0,...,\numexpr\Ymax-1\relax} {
    \draw[edge] (\x,\y) -- (\x,\y+1);
  }
}

\foreach \x in {0,...,\numexpr\C-1\relax} {
  \foreach \y in {0,...,\Ymax} {
    \pgfmathtruncatemacro{\px}{mod(\x,2)}
    \pgfmathtruncatemacro{\py}{mod(\y,2)}
    \ifnum\px=\py
      \draw[edge] (\x,\y) -- (\x+1,\y);
    \fi
  }
}

\foreach \x in {0,...,\C} {
  \foreach \y in {0,...,\Ymax} {
    \node[dot] at (\x,\y) {};
  }
}
\end{tikzpicture}
		\caption{}
		\label{fig:hex}
	\end{subfigure}
	\begin{subfigure}[b]{.22\textwidth}
		\centering
		\begin{tikzpicture}[
  sq/.style={circle, draw=black, line width=0.6pt, fill=white, inner sep=2.2pt},
  ed/.style={black, line cap=round},
  scale =0.43
]
\def\a{2.0}   
\def\Y{3.0}   
\def\b{1.2}   
\def\c{1.2}   

\newcommand{\Sqnode}[3]{\node[sq] (#1) at (#2,#3) {};}

\Sqnode{LM}{0}{0}
\Sqnode{C}{\a}{0}
\Sqnode{RM}{2*\a}{0}
\draw[ed] (LM) -- (C) -- (RM);

\Sqnode{TLc}{0}{\Y}
\Sqnode{TLl}{-\b}{\Y}
\Sqnode{TLr}{ \b}{\Y}
\Sqnode{TLlt}{-\b}{\Y+\c}
\Sqnode{TLlb}{-\b}{\Y-\c}
\Sqnode{TLrt}{ \b}{\Y+\c}
\Sqnode{TLrb}{ \b}{\Y-\c}

\draw[ed] (TLl) -- (TLc) -- (TLr);
\draw[ed] (TLlt) -- (TLl) -- (TLlb);
\draw[ed] (TLrt) -- (TLr) -- (TLrb);
\draw[ed] (TLc) -- (LM);

\Sqnode{TRc}{2*\a}{\Y}
\Sqnode{TRl}{2*\a-\b}{\Y}
\Sqnode{TRr}{2*\a+\b}{\Y}
\Sqnode{TRlt}{2*\a-\b}{\Y+\c}
\Sqnode{TRlb}{2*\a-\b}{\Y-\c}
\Sqnode{TRrt}{2*\a+\b}{\Y+\c}
\Sqnode{TRrb}{2*\a+\b}{\Y-\c}

\draw[ed] (TRl) -- (TRc) -- (TRr);
\draw[ed] (TRlt) -- (TRl) -- (TRlb);
\draw[ed] (TRrt) -- (TRr) -- (TRrb);
\draw[ed] (TRc) -- (RM);

\Sqnode{BLc}{0}{-\Y}
\Sqnode{BLl}{-\b}{-\Y}
\Sqnode{BLr}{ \b}{-\Y}
\Sqnode{BLlt}{-\b}{-\Y+\c}
\Sqnode{BLlb}{-\b}{-\Y-\c}
\Sqnode{BLrt}{ \b}{-\Y+\c}
\Sqnode{BLrb}{ \b}{-\Y-\c}

\draw[ed] (BLl) -- (BLc) -- (BLr);
\draw[ed] (BLlt) -- (BLl) -- (BLlb);
\draw[ed] (BLrt) -- (BLr) -- (BLrb);
\draw[ed] (BLc) -- (LM);

\Sqnode{BRc}{2*\a}{-\Y}
\Sqnode{BRl}{2*\a-\b}{-\Y}
\Sqnode{BRr}{2*\a+\b}{-\Y}
\Sqnode{BRlt}{2*\a-\b}{-\Y+\c}
\Sqnode{BRlb}{2*\a-\b}{-\Y-\c}
\Sqnode{BRrt}{2*\a+\b}{-\Y+\c}
\Sqnode{BRrb}{2*\a+\b}{-\Y-\c}

\draw[ed] (BRl) -- (BRc) -- (BRr);
\draw[ed] (BRlt) -- (BRl) -- (BRlb);
\draw[ed] (BRrt) -- (BRr) -- (BRrb);
\draw[ed] (BRc) -- (RM);

\end{tikzpicture}
		\caption{}
		\label{fig:htree}
	\end{subfigure}
	\caption{Representative two-dimensional qubit layouts: (a) square lattice~\cite{leiserson1980htree}; (b) triangular lattice~\cite{yang2025arxiv}; (c) hexagonal (honeycomb) \cite[FIG. 10d]{cowtan2019tqc} and (d) H-tree layout~\cite{leiserson1980htree}.}
	\label{fig:2dlayout}
\end{figure}

Beyond these practically realizable two-dimensional qubit arrangements, recent studies have also explored higher-dimensional, theoretically motivated qubit topologies intended to reduce depth overhead. Notable examples include the cyclic-butterfly interconnect~\cite{brierly2017qic}, inspired by the recursive Bene\u s network~\cite{benevs1964bstj} and higher dimensional hypercube-based layouts \cite{beals2013rspa}. Although such architectures remain beyond current physical hardware capabilities, they provide valuable insight into upper bounds on resource requirements and guide the design of future scalable quantum architectures. Topologically, these higher-dimensional architectures correspond to higher-degree qubit connectivity. For simplicity, we henceforth refer to all the qubit layouts discussed so far collectively as 2D qubit layouts.

Designing quantum circuits at the logical level and on a 2D qubit layout are fundamentally different. Conventionally, in a logical circuit, qubits are arranged along a vertical line, and gates are applied sequentially from left to right. By contrast, a 2D layout places qubits on a surface, where the circuit evolution is typically depicted from bottom to top.
Moreover, logical circuit design typically assumes an abstract interaction model in which gates can be applied between arbitrary qubits, independent of their physical locations. In graph-theoretic terms, this corresponds to an idealized all-to-all connectivity model, where each qubit may interact with up to ($N_q-1$) others in an ($N_q$)-qubit circuit. In contrast, physical implementations are constrained by the connectivity graph of the hardware, which is often sparse and local in two-dimensional architectures. Therefore, multi-qubit quantum gates between non-adjacent qubits must be supported by additional mapping, routing, or scheduling operations, increasing the overall space-time cost of execution.
Resource estimation and circuit optimization for these 2D qubit layout differ substantially from those in logical circuits, particularly due to increased routing overhead, CNOT map constraints, SWAP-network complexity, and locality-induced depth inflation. Consequently, layout-aware circuit synthesis has become essential for accurately evaluating compilation cost and identifying architecture-specific decompositions.

In this paper, we aim to bridge the gap between logically optimal MCT decompositions from~\cite{dutta2025pra,dutta2026qic} and their physical realizations under 2D qubit connectivity constraints (see TABLE~\ref{tab:mct-comparison}). We undertake a comprehensive analysis of existing decomposition strategies for both basic and multi-controlled Toffoli gates across a range of qubit layouts (including the mapping technique used in IBM Quantum) and highlight the structural features that determine routing overhead. Our results advance the state-of-the-art in architecture-constrained Toffoli and MCT decomposition and contribute to the broader objective of optimizing quantum circuits under physically implementable hardware models.

The resource metrics under consideration are the Toffoli count, Toffoli depth, and Ancilla count. The Toffoli count and Toffoli depth can be further refined with the T-Count and T-Depth, respectively. Moreover, to satisfy connectivity constraints in a finite-size topology, SWAP operations are introduced, and the resulting routing cost is quantified by the SWAP-depth overhead incurred during the mapping.

\subsection{Mapping logical circuits to 2D qubit layouts}
\label{sub:mapping}
Mapping logical quantum circuits onto 2D qubit layouts is a foundational step in quantum compilation~\cite{datta2022mapping, cowtan2019tqc}. These physical topologies typically follow structured lattices, such as the square grids used in Google Sycamore \cite{arute2019quantum} or the heavy-hex layouts of IBM Brisbane~\cite{jurcevic2021qst}. As quantum hardware scales, these architectures are parameterized by size, allowing for theoretically infinite expansion.

One of the primary constraints in 2D layouts is the requirement for physical adjacency between interacting qubits. In architectures with limited connectivity, this constraint is satisfied by inserting SWAP gates to move logical qubits into neighboring positions. Since each SWAP gate decomposes into three alternating CNOT gates, their proliferation increases resource overheads and thus susceptibility to noise. Consequently, minimizing SWAP overhead is critical in architecture-aware circuit mapping \cite{sengupta2023aqucide, zou2025arXiv}.

The mapping problem, i.e., deciding where to place the qubits and how to route required interactions, is computationally challenging. Optimal placement and routing are NP-hard \cite{bhattacharjee2017arxiv, wille2019dac}, prompting extensive work on bounding the depth overhead for various topologies. Theoretical approaches often treat qubit routing as a graph-permutation problem. Alon et al.~\cite{alon1994siam} characterise the routing problem in this way and provide upper bounds on the depth of the routing circuit (routing number $rt(G)$) for several graph families using matching, including paths($P_n)$, grids, and hypercubes($Q_r$). Childs et al.~\cite{childs2019tqc} extend this matching-based methodology for routing in quantum layouts to derive upper bounds on depth overhead.

An alternative line of work employs sorting-based techniques. By representing qubit movement as a sorting problem, this method determines SWAP paths for long-range interactions by treating the underlying topology as a specialized sorting network. In this way, it provides upper bounds on the SWAP depth for architectures with limited connectivity \cite{beals2013rspa}. This approach has produced competitive bounds on several topologies, with the cyclic butterfly network achieving an overhead of $6\log_2 n$ \cite{brierly2017qic}.

In practical applications, frameworks like IBM’s Qiskit manage these constraints using coupling maps, graph representations of hardware connectivity. The standard mapping tool in Qiskit is SABRE (SWAP-Aware Buffer Re-arrangement) \cite{li2019asplos}, a heuristic-based algorithm that routes logical qubits onto the physical layout. However, these heuristics are often non-deterministic and may not exploit specific gate symmetries or topological redundancies.

To address these limitations, this work establishes two primary motivations:
\begin{itemize}
	\item \textbf{Topology-aware placement on infinite grid architecture}: The SWAP overhead for an $n$-MCT gate strongly depends on the available qubit connectivity. We show that by allowing sufficiently large topologies, an $n$-MCT gate can be embedded without incurring additional depth overhead. This leads to a topology-driven placement strategy that eliminates routing costs using the required ancilla.
	\item \textbf{Bounding depth overhead under minimal topology-size constraints}: Given the non-deterministic nature of heuristic mappers, we propose a method to bound the additional depth, based on minimum qubit level constraints. By defining the minimum required topology size $$r_{\min}= \operatorname*{min}_{r \in \mathbb{N}} \{ r \mid \mathrm{size}(r) \ge q_{\mathrm{req}} \},$$
	Where $q_{req}$ denotes the total qubits that are mapped to the hardware. Through this constraint, we provide an architecture-aware placement strategy that ensures predictable performance on known 2D layouts.
\end{itemize}
The following subsection provides a brief summary of the earlier works in this direction.

\subsection{Related works}
\label{sub:related work}
Implementing quantum circuits on 2D qubit layouts has been an active research direction for several decades, with multiple approaches explored from different perspectives. 

In 2010, Saeedi et al.~\cite{saeedi2011qinp} extended existing synthesis flows to obtain Linear Nearest-Neighbor (LNN)-compliant circuits. Their method introduced template-matching optimizations, an exact synthesis strategy, and a couple of reordering heuristics, achieving an average improvement of over 50\% in quantum cost. In 2011, Hirata et al.~\cite{hirata2011qic} conjectured that, given an arbitrary quantum circuit, finding an LNN realization with the minimum number of SWAP gates is NP-complete. In 2013, Shafaei et al.~\cite{shafaei2013dac} employed a Minimum Linear Arrangement (MinLA) formulation on the qubit interaction graph for local optimization. Their approach partitions circuits into subcircuits, applies MinLA-based reordering within each part, and inserts SWAP operations both locally and between partitions to ensure functional equivalence. In 2014, Wille et al.~\cite{wille2014dac} proposed an optimal SWAP insertion strategy using pseudo-Boolean optimization to enforce nearest-neighbor constraints. Additionally, in 2018, Zulehner et al.~\cite{zulehner2018date} developed an architecture-aware mapping framework for IBM QX architectures combining A* search, look-ahead heuristics, and depth-based partitioning. This approach, later integrated into Qiskit, significantly reduces mapping time while minimizing gate and depth overhead, often producing results within seconds compared to the hour-scale runtimes of IBM's original solution.

In parallel, numerous works have studied 2D implementations of key quantum subroutines, including the quantum Fourier transform~\cite{takahashi2007qic, maslov2007pra}, Shor's algorithm~\cite{fowler2004qic, pham2013qic}, quantum error correction~\cite{fowler2004pra}, and so on. Several attempts have also focused on implementing Toffoli, MCT, and related gates under restricted 2D layouts. In 2019, Hu et al.~\cite{hu2019dac} proposed an efficient MCT decomposition on 2D square lattices using fractal H-tree layouts and relative-phase Toffoli gates~\cite{maslov2016pra}. They demonstrated MCT implementations with up to 12 controls without SWAP overhead and achieved circuit depths lower than the standard baseline~\cite{nc}. In 2021, Matsuo et al.~\cite{matsuo2021aspdac} introduced a framework for dynamically decomposing mixed-polarity MCT (MPMCT) gates, integrating SWAP insertion directly into the decomposition step rather than treating routing and decomposition sequentially. In 2023, Paler et al.~\cite{paler2024nanoarch} presented a Clifford+T realization of the Toffoli gate with T-depth 2 and zero SWAP overhead on 2D square lattices.

Multiple studies have considered mapping complex quantum gates within IBM-specific coupling architectures. In 2020, Niemann et al.~\cite{niemann2020dsd} investigated the mapping of reversible (MCT-based) circuits onto the IBM QX5 (R\"uschlikon) device, showing that hardware-aware decompositions combined with circuit-structural insights can substantially reduce overhead. They also presented a 2D mapping of the T-depth-3 Toffoli decomposition and reported empirical runtimes on IBM QX5. Later, Niemann et al.~\cite{niemann2021date} demonstrated that combining SWAP operations with remote-Toffoli constructions can further reduce mapping costs. Their study compared resource estimates for basic and multi-controlled Toffoli gates across IBM QX5 and IBM Q20 architectures. However, most of these constructions predate recent advances in Toffoli and MCT decompositions and therefore require significant revision. 

In this paper, we implement state-of-the-art Toffoli and MCT decompositions from~\cite[TABLEs I and II]{dutta2025pra} across various 2D qubit layouts and provide updated resource estimates. 

\subsection{Organization and contributions}
\label{sub:org}
Section~\ref{sec:warmup} serves as a warm-up by presenting 2D mappings of existing basic Toffoli decompositions on various qubit layouts. For each decomposition, the chosen layout matches the corresponding qubit-interaction graph, thereby eliminating any additional resource overhead when mapping the logical circuit to its 2D implementation.

In Section~\ref{sec:cont1}, we investigate the mapping of optimal Toffoli-depth MCT decomposition techniques onto various 2D layouts, assuming infinitely scalable topologies and hence incurring no additional depth overhead. Specifically, we consider triangular, square-grid, square-with-diagonal, and H-tree layouts, and derive the corresponding quantum resource requirements for implementing an $n$-MCT gate.

Section~\ref{sec:depth_bound} derives upper bounds on the depth overhead for an architecture-aware placement technique under minimum-size hardware constraints. We model the local interaction graphs induced by various Toffoli decompositions as graph motifs and embed them optimally within constrained topologies. The resulting MCT decomposition layers are then mapped onto 2D square-grid and hypercube layouts. These bounds quantify the additional depth overhead incurred by motif packing under tight qubit budgets.

In Section~\ref{sec:experiments}, we validate the proposed bounds using routing-aware heuristics. As an additional assessment, we evaluate how efficiently different motifs can be packed into various hardware topologies. In particular, we report matching-rate results for several MIS-based packing heuristics to assess motif-embedding quality. 

Finally, Section~\ref{sec:con} concludes the paper by summarizing our contributions and outlining several open problems for future research.

\section{Warm up: Mapping basic Toffoli decompositions onto 2D Qubit Layouts Without SWAP}
\label{sec:warmup}
In this section, we examine mapping strategies for the basic Toffoli decompositions from~\cite[TABLE I]{dutta2025pra} onto various 2D qubit layouts, including square-grid, triangular, hexagonal, and H-tree topologies, without introducing SWAP operations to satisfy CNOT connectivity. Since our primary objective here is to eliminate SWAP-induced depth overhead, we manually analyze the CNOT connectivity requirements of each decomposition and identify the corresponding 2D layouts that support them natively.

We begin with the 2D mapping of the 7-T Toffoli decompositions from~\cite{nc, amy2013ieee, selinger2013pra}, followed by state-of-the-art decompositions employing measurement-based uncomputation using 4 T gates. A summary of the resulting resource estimates for all considered layouts is provided in TABLE~\ref{tab:tof-2d}.

\subsection{Basic Toffoli decomposition without ancilla on triangular lattice topology}
\label{sub:triangle}
In~\cite[Chap. 4, Sec. 4.3]{nc}, Nielsen and Chuang presented a Clifford+T decomposition of the basic Toffoli gate using seven T gates, with a T-depth of 6. In this logical decomposition (FIG.~\ref{fig:NCTD6}), all three qubits are mutually connected through CNOT operations. Consequently, to avoid SWAP overhead when mapping the logical circuit to a 2D qubit layout, both the two control qubits ($c_1, c_2$) and the target qubit ($t$) must be placed adjacent to one another, naturally requiring a triangular layout.
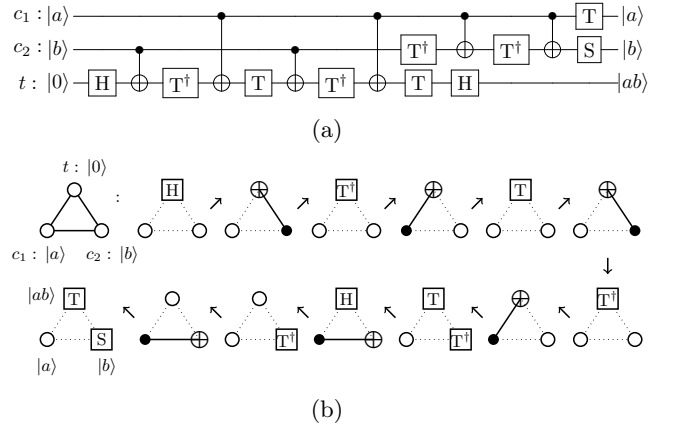
\begin{figure}[htbp]
	\centering
	\begin{subfigure}{0.48\textwidth}
		\centering
		\includegraphics[width=1\linewidth]{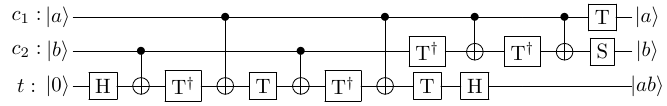}
		\caption{}
		\label{fig:NCTD6}
	\end{subfigure}\\
	\begin{subfigure}{0.48\textwidth}
		\centering
		\scalebox{0.72}{\begin{tikzpicture}[
open/.style={circle, draw=black, fill = white, thick, minimum size=2.5mm, inner sep=0pt},
open_nt/.style={fill = white, thick, minimum size=3mm, inner sep=0pt},
open_cn/.style={circle, minimum size=2.5mm, inner sep=0pt},
open_ctr/.style={circle, minimum size=2mm,fill = black, inner sep=0pt},
gate_sq/.style={rectangle, draw=black, fill = white, thick, minimum size=3.8mm, inner sep=0pt},
filled/.style={circle, draw=black, thick, fill=black, minimum size=5mm, inner sep=0pt},
  op/.style={circle, draw=black, thick, minimum size=6mm, inner sep=0pt, font=\small},
  edge/.style={black, thick},
  edge_d/.style={black, thick,dashed},
  gate/.style={draw=black, thick, minimum width=6mm, minimum height=6mm, inner sep=1pt, font=\small},
  arr/.style={->, thick}
]

\newcommand{\TriangleGraph}[5]{%
  \begin{scope}[shift={(#1,#5)}]
    \coordinate (c1) at (0,0);
    \coordinate (c2) at (1,0);
    \coordinate (t) at (0.5,0.75);

    \draw[edge] (c1)--(c2);
    \draw[edge] (c1)--(t);
    \draw[edge] (c2)--(t);

    \node[open] (ntl) at (c1) {#2};
    \node[open] (ntr) at (c2) {#3};
    \node[open] (nbl) at (t) {#4};
  \end{scope}
}


\TriangleGraph{0}{}{}{}{0}
\node[font=\small] at (-0.15,-0.45) {$c_1: \,|a\rangle$};
\node[font=\small] at (1.2,-0.45) {$c_2:\,|b\rangle$};
\node[font=\small] at (0.7,1.15) {$t:\,|0\rangle$};
%
\node[] at (1.3, 0.6) {$\mathrm{:}$};

\begin{scope}[shift={(1.8,0)}]
    \coordinate (c1) at (0,0);
    \coordinate (c2) at (1,0);
    \coordinate (t) at (0.5,0.75);

    \draw[dotted] (c1)--(c2);
    \draw[dotted] (c1)--(t);
    \draw[dotted] (c2)--(t);

    \node[open] (nbl) at (c1){};
    \node[open] (nbl) at (c2){};
    \node[gate_sq] (nbl) at (t){$\mathrm{H}$};

\end{scope}


\draw[arr] (3,0.4) -- (3.2,0.6);

\begin{scope}[shift={(3.4,0)}]
    \coordinate (c1) at (0,0);
    \coordinate (c2) at (1,0);
    \coordinate (t) at (0.5,0.75);

    \draw[dotted] (c1)--(c2);
    \draw[edge] (c2)--(t);
    \draw[dotted] (c1)--(t);

    \node[open_cn] (ntr) at (t) {$\bigoplus$};
    \node[open] (nbl) at (c1) {};
    \node[open_ctr] (ntl) at (c2) {};
\end{scope}

\draw[arr] (4.6,0.4) -- (4.8,0.6);

\begin{scope}[shift={(5,0)}]
    \coordinate (c1) at (0,0);
    \coordinate (c2) at (1,0);
    \coordinate (t) at (0.5,0.75);

    \draw[dotted] (c1)--(c2);
    \draw[dotted] (c1)--(t);
    \draw[dotted] (c2)--(t);

    \node[open] (ntl) at (c1) {};
    \node[open] (ntr) at (c2) {};
    \node[gate_sq] (nbl) at (t) {$\mathrm{T}^{\dagger}$};

\end{scope}

\draw[arr] (6.2,0.4) -- (6.4,0.6);

\begin{scope}[shift={(6.6,0)}]
    \coordinate (c1) at (0,0);
    \coordinate (c2) at (1,0);
    \coordinate (t) at (0.5,0.75);

    \draw[edge] (c1)--(t);
    \draw[dotted] (c1)--(c2);
    \draw[dotted] (c2)--(t);

    \node[open_cn] (ntl) at (t) {$\bigoplus$};
    \node[open_ctr] (ntr) at (c1) {};
    \node[open] (nbl) at (c2) {};

\end{scope}

\draw[arr] (7.8,0.4) -- (8,0.6) ;


\begin{scope}[shift={(8.2,0.0)}]
    \coordinate (c1) at (0,0);
    \coordinate (c2) at (1,0);
    \coordinate (t) at (0.5,0.75);

    \draw[dotted] (c1)--(c2);
    \draw[dotted] (c1)--(t);
    \draw[dotted] (c2)--(t);

    \node[open] (nbl) at (c1) {};
    \node[open] (ntr) at (c2) {};
    \node[gate_sq] (nbl) at (t) {$\mathrm{T}$};
\end{scope}
\draw[arr] (9.4,0.4) -- (9.6,0.6);

\begin{scope}[shift={(9.8,0.0)}]
    \coordinate (c1) at (0,0);
    \coordinate (c2) at (1,0);
    \coordinate (t) at (0.5,0.75);

    \draw[dotted] (c1)--(c2);
    \draw[dotted] (c1)--(t);
    \draw[edge] (c2)--(t);

    \node[open_cn] (ntl) at (t) {$\bigoplus$};
    \node[open_ctr] (ntr) at (c2) {};
    \node[open] (nbl) at (c1) {};
\end{scope}

\draw[arr] (10.3,-0.5) -- (10.3,-0.8);

\begin{scope}[shift={(9.8,-2.0)}]
    \coordinate (c1) at (0,0);
    \coordinate (c2) at (1,0);
    \coordinate (t) at (0.5,0.75);

    \draw[dotted] (c1)--(c2);
    \draw[dotted] (c1)--(t);
    \draw[dotted] (c2)--(t);

    \node[open] (ntl) at (c1) {};
    \node[open] (ntr) at (c2) {};
    \node[gate_sq] (nbl) at (t) {$\mathrm{T}^{\dagger}$};
\end{scope}

\draw[arr] (9.6,-1.6) -- (9.4,-1.4);

\begin{scope}[shift={(8.2,-2.0)}]
    \coordinate (c1) at (0,0);
    \coordinate (c2) at (1,0);
    \coordinate (t) at (0.5,0.75);

    \draw[dotted] (c1)--(c2);
    \draw[edge] (c1)--(t);
    \draw[dotted] (c2)--(t);

    \node[open_ctr] (ntl) at (c1) {};
    \node[open] (ntr) at (c2) {};
    \node[open_cn] (nbl) at (t) {$\bigoplus$};
\end{scope}

\draw[arr] (8,-1.6) -- (7.8,-1.4);

\begin{scope}[shift={(6.6,-2.0)}]
    \coordinate (c1) at (0,0);
    \coordinate (c2) at (1,0);
    \coordinate (t) at (0.5,0.75);

    \draw[dotted] (c1)--(c2);
    \draw[dotted] (c1)--(t);
    \draw[dotted] (c2)--(t);

    \node[gate_sq] (nbl) at (t) {$\mathrm{T}$};
    \node[open] (ntl) at (c1) {};
    \node[gate_sq] (ntr) at (c2) {$\mathrm{T}^{\dagger}$};
\end{scope}

\draw[arr] (6.4,-1.6) -- (6.2,-1.4);

\begin{scope}[shift={(5.0,-2.0)}]
    \coordinate (c1) at (0,0);
    \coordinate (c2) at (1,0);
    \coordinate (t) at (0.5,0.75);

    \draw[edge] (c1)--(c2);
    \draw[dotted] (c1)--(t);
    \draw[dotted] (c2)--(t);

    \node[open_ctr] (ntl) at (c1) {};
    \node[open_cn] (ntr) at (c2) {$\bigoplus$};
    \node[gate_sq] (nbl) at (t) {$\mathrm{H}$};
\end{scope}

\draw[arr] (4.8,-1.6) -- (4.6,-1.4);

\begin{scope}[shift={(3.4,-2.0)}]

    \coordinate (c1) at (0,0);
    \coordinate (c2) at (1,0);
    \coordinate (t) at (0.5,0.75);

    \draw[dotted] (c1)--(c2);
    \draw[dotted] (c1)--(t);
    \draw[dotted] (c2)--(t);

    \node[open] (ntl) at (c1) {};
    \node[gate_sq] (ntr) at (c2) {$\mathrm{T}^{\dagger}$};
    \node[open] (nbl) at (t) {};
\end{scope}

\draw[arr] (3.2,-1.6) -- (3.0,-1.4);

\begin{scope}[shift={(1.8,-2.0)}]

    \coordinate (c1) at (0,0);
    \coordinate (c2) at (1,0);
    \coordinate (t) at (0.5,0.75);

    \draw[edge] (c1)--(c2);
    \draw[dotted] (c1)--(t);
    \draw[dotted] (c2)--(t);

    \node[open_ctr] (ntl) at (c1) {};
    \node[open_cn] (ntr) at (c2) {$\bigoplus$};
    \node[open] (nbl) at (t) {};
\end{scope}

\draw[arr] (1.6,-1.6) -- (1.4,-1.4);

\begin{scope}[shift={(0,-2.0)}]

    \coordinate (c1) at (0,0);
    \coordinate (c2) at (1,0);
    \coordinate (t) at (0.5,0.75);

    \draw[dotted] (c1)--(c2);
    \draw[dotted] (c1)--(t);
    \draw[dotted] (c2)--(t);

    \node[open] (ntl) at (c1) {};
    \node[gate_sq] (ntr) at (c2) {$\mathrm{S}$};
    \node[gate_sq] (nbl) at (t) {$\mathrm{T}$};
\end{scope}

\node[] at (0, -2.5) {$|a\rangle$};
\node[] at (1.1,-2.5) {$|b\rangle$};
\node[] at (-0.1, -1.2) {$|ab\rangle$};

\end{tikzpicture}}
		\caption{}
		\label{fig:NCTD6-2d}
	\end{subfigure}
	\caption{(a) T-depth 6 Toffoli decomposition from~\cite{nc}, and (b) its corresponding 2D implementation.}
	\label{fig:td6}
\end{figure}
The resulting nearest-neighbor circuit (FIG.~\ref{fig:NCTD6-2d}) consists of 13 sequential layers on top of each other, implying a circuit depth of 13, of which 6 contain T gates, giving a T-depth of 6. Because the mapping topology was selected based on the qubit interactions, no SWAP gate is required.

Similarly, in 2013, Amy et al.~\cite{amy2013ieee} introduced two Toffoli decompositions, each using seven T gates, requiring no ancilla qubits, and achieving T-depths of 3 (FIG.~\ref{fig:AmyTD3}) and 4 (FIG.~\ref{fig:AmyTD4}). In both these constructions, all three qubits interact through one or more CNOT gates, implying that every pair of qubits must be connected. This enforces degree-2 connectivity among the three qubits, again corresponding to a triangular-lattice topology.
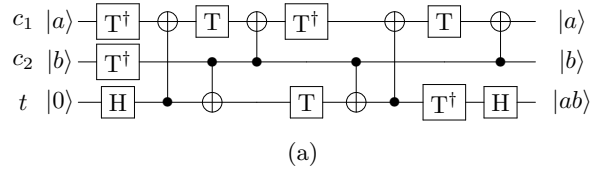
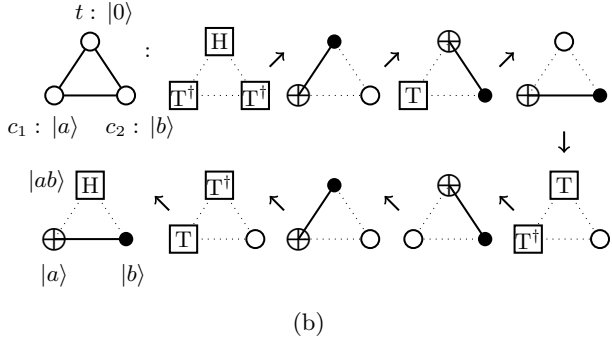
\begin{figure}[ht]
	\centering
	\begin{subfigure}{0.48\textwidth}
		\centering
		\includegraphics[width=0.9\linewidth]{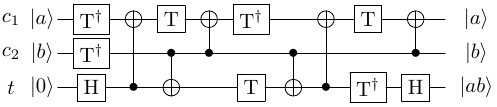}
		\caption{}
		\label{fig:AmyTD4}
	\end{subfigure}\\
	\begin{subfigure}{0.48\textwidth}
		\centering
		\scalebox{0.95}{\begin{tikzpicture}[
open/.style={circle, draw=black, fill = white, thick, minimum size=2.5mm, inner sep=0pt},
open_nt/.style={fill = white, thick, minimum size=3mm, inner sep=0pt},
open_cn/.style={circle, minimum size=2.5mm, inner sep=0pt},
open_ctr/.style={circle, minimum size=2mm,fill = black, inner sep=0pt},
gate_sq/.style={rectangle, draw=black, fill = white, thick, minimum size=3.8mm, inner sep=0pt},
filled/.style={circle, draw=black, thick, fill=black, minimum size=5mm, inner sep=0pt},
  op/.style={circle, draw=black, thick, minimum size=6mm, inner sep=0pt, font=\small},
  edge/.style={black, thick},
  edge_d/.style={black, thick,dashed},
  gate/.style={draw=black, thick, minimum width=6mm, minimum height=6mm, inner sep=1pt, font=\small},
  arr/.style={->, thick}
]

\newcommand{\TriangleGraph}[5]{%
  \begin{scope}[shift={(#1,#5)}]
    \coordinate (c1) at (0,0);
    \coordinate (c2) at (1,0);
    \coordinate (t) at (0.5,0.75);

    \draw[edge] (c1)--(c2);
    \draw[edge] (c1)--(t);
    \draw[edge] (c2)--(t);

    \node[open] (ntl) at (c1) {#2};
    \node[open] (ntr) at (c2) {#3};
    \node[open] (nbl) at (t) {#4};
  \end{scope}
}


\TriangleGraph{0}{}{}{}{0}
\node[font=\small] at (-0.15,-0.45) {$c_1: \,|a\rangle$};
\node[font=\small] at (1.2,-0.45) {$c_2:\,|b\rangle$};
\node[font=\small] at (0.7,1.15) {$t:\,|0\rangle$};
%
\node[] at (1.3, 0.6) {$\mathrm{:}$};

\begin{scope}[shift={(1.8,0)}]
    \coordinate (c1) at (0,0);
    \coordinate (c2) at (1,0);
    \coordinate (t) at (0.5,0.75);

    \draw[dotted] (c1)--(c2);
    \draw[dotted] (c1)--(t);
    \draw[dotted] (c2)--(t);

    \node[gate_sq] (nbl) at (c1){$\mathrm{T}^{\dagger}$};
    \node[gate_sq] (nbl) at (c2){$\mathrm{T}^{\dagger}$};
    \node[gate_sq] (nbl) at (t){$\mathrm{H}$};

\end{scope}


\draw[arr] (3,0.4) -- (3.2,0.6);

\begin{scope}[shift={(3.4,0)}]
    \coordinate (c1) at (0,0);
    \coordinate (c2) at (1,0);
    \coordinate (t) at (0.5,0.75);

    \draw[dotted] (c1)--(c2);
    \draw[edge] (c1)--(t);
    \draw[dotted] (c2)--(t);

    \node[open_ctr] (ntr) at (t) {};
    \node[open_cn] (nbl) at (c1) {$\bigoplus$};
    \node[open] (ntl) at (c2) {};
\end{scope}

\draw[arr] (4.6,0.4) -- (4.8,0.6);

\begin{scope}[shift={(5,0)}]
    \coordinate (c1) at (0,0);
    \coordinate (c2) at (1,0);
    \coordinate (t) at (0.5,0.75);

    \draw[dotted] (c1)--(c2);
    \draw[dotted] (c1)--(t);
    \draw[edge] (c2)--(t);

    \node[gate_sq] (ntl) at (c1) {$\mathrm{T}$};
    \node[open_ctr] (ntr) at (c2) {};
    \node[open_cn] (nbl) at (t) {$\bigoplus$};

\end{scope}

\draw[arr] (6.2,0.4) -- (6.4,0.6);

\begin{scope}[shift={(6.6,0)}]
    \coordinate (c1) at (0,0);
    \coordinate (c2) at (1,0);
    \coordinate (t) at (0.5,0.75);

    \draw[edge] (c1)--(c2);
    \draw[dotted] (c1)--(t);
    \draw[dotted] (c2)--(t);

    \node[open] (ntl) at (t) {};
    \node[open_cn] (ntr) at (c1) {$\bigoplus$};
    \node[open_ctr] (nbl) at (c2) {};

\end{scope}

\draw[arr] (7.1,-0.5) -- (7.1,-0.8) ;

\begin{scope}[shift={(6.6,-2.0)}]
    \coordinate (c1) at (0,0);
    \coordinate (c2) at (1,0);
    \coordinate (t) at (0.5,0.75);

    \draw[dotted] (c1)--(c2);
    \draw[dotted] (c1)--(t);
    \draw[dotted] (c2)--(t);

    \node[gate_sq] (nbl) at (t) {$\mathrm{T}$};
    \node[open] (ntr) at (c2) {};
    \node[gate_sq] (ntl) at (c1) {$\mathrm{T}^{\dagger}$};
\end{scope}

\draw[arr] (6.4,-1.6) -- (6.2,-1.4);

\begin{scope}[shift={(5.0,-2.0)}]
    \coordinate (c1) at (0,0);
    \coordinate (c2) at (1,0);
    \coordinate (t) at (0.5,0.75);

    \draw[edge] (t)--(c2);
    \draw[dotted] (c1)--(t);
    \draw[dotted] (c2)--(c1);

    \node[open] (ntl) at (c1) {};
    \node[open_ctr] (ntr) at (c2) {};
    \node[open_cn] (nbl) at (t) {$\bigoplus$};
\end{scope}

\draw[arr] (4.8,-1.6) -- (4.6,-1.4);

\begin{scope}[shift={(3.4,-2.0)}]

    \coordinate (c1) at (0,0);
    \coordinate (c2) at (1,0);
    \coordinate (t) at (0.5,0.75);

    \draw[dotted] (c1)--(c2);
    \draw[edge] (c1)--(t);
    \draw[dotted] (c2)--(t);

    \node[open_cn] (ntl) at (c1) {$\bigoplus$};
    \node[open] (ntr) at (c2) {};
    \node[open_ctr] (nbl) at (t) {};
\end{scope}

\draw[arr] (3.2,-1.6) -- (3.0,-1.4);

\begin{scope}[shift={(1.8,-2.0)}]

    \coordinate (c1) at (0,0);
    \coordinate (c2) at (1,0);
    \coordinate (t) at (0.5,0.75);

    \draw[dotted] (c1)--(c2);
    \draw[dotted] (c1)--(t);
    \draw[dotted] (c2)--(t);

    \node[gate_sq] (ntl) at (c1) {$\mathrm{T}$};
    \node[open] (ntr) at (c2) {};
    \node[gate_sq] (nbl) at (t) {$\mathrm{T}^{\dagger}$};
\end{scope}

\draw[arr] (1.6,-1.6) -- (1.4,-1.4);

\begin{scope}[shift={(0,-2.0)}]

    \coordinate (c1) at (0,0);
    \coordinate (c2) at (1,0);
    \coordinate (t) at (0.5,0.75);

    \draw[edge] (c1)--(c2);
    \draw[dotted] (c1)--(t);
    \draw[dotted] (c2)--(t);

    \node[open_cn] (ntl) at (c1) {$\bigoplus$};
    \node[open_ctr] (ntr) at (c2) {};
    \node[gate_sq] (nbl) at (t) {$\mathrm{H}$};
\end{scope}

\node[] at (0, -2.5) {$|a\rangle$};
\node[] at (1.1,-2.5) {$|b\rangle$};
\node[] at (-0.1, -1.2) {$|ab\rangle$};

\end{tikzpicture}}
		\caption{}
		\label{fig:AmyTD4-2d}
	\end{subfigure}
	\caption{(a) T-depth 4 Toffoli decomposition from~\cite{amy2013ieee}, and (b) its corresponding 2D implementation.}
	\label{fig:td4}
\end{figure}
\begin{figure}[ht]
	\centering
	\begin{subfigure}{0.48\textwidth}
		\centering
		\includegraphics[width=0.9\linewidth]{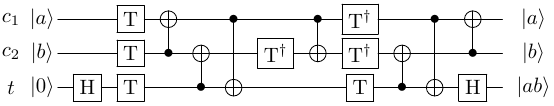}
		\caption{}
		\label{fig:AmyTD3}
	\end{subfigure}\\
	\begin{subfigure}{0.48\textwidth}
		\centering
		\scalebox{0.84}{\begin{tikzpicture}[
open/.style={circle, draw=black, fill = white, thick, minimum size=2.5mm, inner sep=0pt},
open_nt/.style={fill = white, thick, minimum size=3mm, inner sep=0pt},
open_cn/.style={circle, minimum size=2.5mm, inner sep=0pt},
open_ctr/.style={circle, minimum size=2mm,fill = black, inner sep=0pt},
gate_sq/.style={rectangle, draw=black, fill = white, thick, minimum size=3.8mm, inner sep=0pt},
filled/.style={circle, draw=black, thick, fill=black, minimum size=5mm, inner sep=0pt},
  op/.style={circle, draw=black, thick, minimum size=6mm, inner sep=0pt, font=\small},
  edge/.style={black, thick},
  edge_d/.style={black, thick,dashed},
  gate/.style={draw=black, thick, minimum width=6mm, minimum height=6mm, inner sep=1pt, font=\small},
  arr/.style={->, thick}
]

\newcommand{\TriangleGraph}[5]{%
  \begin{scope}[shift={(#1,#5)}]
    \coordinate (c1) at (0,0);
    \coordinate (c2) at (1,0);
    \coordinate (t) at (0.5,0.75);

    \draw[edge] (c1)--(c2);
    \draw[edge] (c1)--(t);
    \draw[edge] (c2)--(t);

    \node[open] (ntl) at (c1) {#2};
    \node[open] (ntr) at (c2) {#3};
    \node[open] (nbl) at (t) {#4};
  \end{scope}
}


\TriangleGraph{0}{}{}{}{0}
\node[font=\small] at (-0.15,-0.45) {$c_1: \,|a\rangle$};
\node[font=\small] at (1.2,-0.45) {$c_2:\,|b\rangle$};
\node[font=\small] at (0.7,1.15) {$t:\,|0\rangle$};
%
\node[] at (1.3, 0.6) {$\mathrm{:}$};

\begin{scope}[shift={(1.8,0)}]
    \coordinate (c1) at (0,0);
    \coordinate (c2) at (1,0);
    \coordinate (t) at (0.5,0.75);

    \draw[dotted] (c1)--(c2);
    \draw[dotted] (c1)--(t);
    \draw[dotted] (c2)--(t);

    \node[open] (nbl) at (c1){};
    \node[open] (nbl) at (c2){};
    \node[gate_sq] (nbl) at (t){$\mathrm{H}$};

\end{scope}


\draw[arr] (3,0.4) -- (3.2,0.6);

\begin{scope}[shift={(3.4,0)}]
    \coordinate (c1) at (0,0);
    \coordinate (c2) at (1,0);
    \coordinate (t) at (0.5,0.75);

    \draw[dotted] (c1)--(c2);
    \draw[dotted] (c2)--(t);
    \draw[dotted] (c1)--(t);

    \node[gate_sq] (ntr) at (t) {$\mathrm{T}$};
    \node[gate_sq] (nbl) at (c1) {$\mathrm{T}$};
    \node[gate_sq] (ntl) at (c2) {$\mathrm{T}$};
\end{scope}

\draw[arr] (4.6,0.4) -- (4.8,0.6);

\begin{scope}[shift={(5,0)}]
    \coordinate (c1) at (0,0);
    \coordinate (c2) at (1,0);
    \coordinate (t) at (0.5,0.75);

    \draw[edge] (c1)--(c2);
    \draw[dotted] (c1)--(t);
    \draw[dotted] (c2)--(t);

    \node[open_cn] (ntl) at (c1) {$\bigoplus$};
    \node[open_ctr] (ntr) at (c2) {};
    \node[open] (nbl) at (t) {};

\end{scope}

\draw[arr] (6.2,0.4) -- (6.4,0.6);

\begin{scope}[shift={(6.6,0)}]
    \coordinate (c1) at (0,0);
    \coordinate (c2) at (1,0);
    \coordinate (t) at (0.5,0.75);

    \draw[edge] (c2)--(t);
    \draw[dotted] (c1)--(c2);
    \draw[dotted] (c1)--(t);

    \node[open_ctr] (ntl) at (t) {};
    \node[open] (ntr) at (c1) {};
    \node[open_cn] (nbl) at (c2) {$\bigoplus$};

\end{scope}

\draw[arr] (7.8,0.4) -- (8,0.6) ;


\begin{scope}[shift={(8.2,0.0)}]
    \coordinate (c1) at (0,0);
    \coordinate (c2) at (1,0);
    \coordinate (t) at (0.5,0.75);

    \draw[dotted] (c1)--(c2);
    \draw[edge] (c1)--(t);
    \draw[dotted] (c2)--(t);

    \node[open_ctr] (nbl) at (c1) {};
    \node[open] (ntr) at (c2) {};
    \node[open_cn] (nbl) at (t) {$\bigoplus$};
\end{scope}

\draw[arr] (8.7,-0.5) -- (8.7,-0.8);

\begin{scope}[shift={(8.2,-2.0)}]
    \coordinate (c1) at (0,0);
    \coordinate (c2) at (1,0);
    \coordinate (t) at (0.5,0.75);

    \draw[dotted] (c1)--(c2);
    \draw[dotted] (c1)--(t);
    \draw[dotted] (c2)--(t);

    \node[open] (ntl) at (c1) {};
    \node[gate_sq] (ntr) at (c2) {$\mathrm{T}^{\dagger}$};
    \node[open] (nbl) at (t) {};
\end{scope}

\draw[arr] (8,-1.6) -- (7.8,-1.4);

\begin{scope}[shift={(6.6,-2.0)}]
    \coordinate (c1) at (0,0);
    \coordinate (c2) at (1,0);
    \coordinate (t) at (0.5,0.75);

    \draw[edge] (c1)--(c2);
    \draw[dotted] (c1)--(t);
    \draw[dotted] (c2)--(t);

    \node[open] (nbl) at (t) {};
    \node[open_ctr] (ntl) at (c1) {};
    \node[open_cn] (ntr) at (c2) {$\bigoplus$};
\end{scope}

\draw[arr] (6.4,-1.6) -- (6.2,-1.4);

\begin{scope}[shift={(5.0,-2.0)}]
    \coordinate (c1) at (0,0);
    \coordinate (c2) at (1,0);
    \coordinate (t) at (0.5,0.75);

    \draw[dotted] (c1)--(c2);
    \draw[dotted] (c1)--(t);
    \draw[dotted] (c2)--(t);

    \node[gate_sq] (ntl) at (c1) {$\mathrm{T}^{\dagger}$};
    \node[gate_sq] (ntr) at (c2) {$\mathrm{T}^{\dagger}$};
    \node[gate_sq] (nbl) at (t) {$\mathrm{T}$};
\end{scope}

\draw[arr] (4.8,-1.6) -- (4.6,-1.4);

\begin{scope}[shift={(3.4,-2.0)}]

    \coordinate (c1) at (0,0);
    \coordinate (c2) at (1,0);
    \coordinate (t) at (0.5,0.75);

    \draw[dotted] (c1)--(c2);
    \draw[dotted] (c1)--(t);
    \draw[edge] (c2)--(t);

    \node[open] (ntl) at (c1) {};
    \node[open_cn] (ntr) at (c2) {$\bigoplus$};
    \node[open_ctr] (nbl) at (t) {};
\end{scope}

\draw[arr] (3.2,-1.6) -- (3.0,-1.4);

\begin{scope}[shift={(1.8,-2.0)}]

    \coordinate (c1) at (0,0);
    \coordinate (c2) at (1,0);
    \coordinate (t) at (0.5,0.75);

    \draw[edge] (c1)--(t);
    \draw[dotted] (c1)--(c2);
    \draw[dotted] (c2)--(t);

    \node[open_ctr] (ntl) at (c1) {};
    \node[open] (ntr) at (c2) {};
    \node[open_cn] (nbl) at (t) {$\bigoplus$};
\end{scope}

\draw[arr] (1.6,-1.6) -- (1.4,-1.4);

\begin{scope}[shift={(0,-2.0)}]

    \coordinate (c1) at (0,0);
    \coordinate (c2) at (1,0);
    \coordinate (t) at (0.5,0.75);

    \draw[edge] (c1)--(c2);
    \draw[dotted] (c1)--(t);
    \draw[dotted] (c2)--(t);

    \node[open_cn] (ntl) at (c1) {$\bigoplus$};
    \node[open_ctr] (ntr) at (c2) {};
    \node[gate_sq] (nbl) at (t) {$\mathrm{H}$};
\end{scope}

\node[] at (0, -2.5) {$|a\rangle$};
\node[] at (1.1,-2.5) {$|b\rangle$};
\node[] at (-0.1, -1.2) {$|ab\rangle$};

\end{tikzpicture}}
		\caption{}
		\label{fig:AmyTD3-2d}
	\end{subfigure}
	\caption{(a) T-depth 3 Toffoli decomposition from~\cite{amy2013ieee}, and (b) its corresponding 2D implementation.}
	\label{fig:td3}
\end{figure}
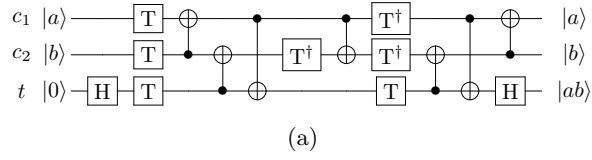
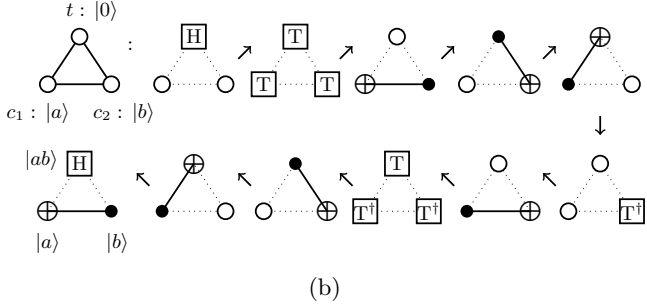
The nearest-neighbor circuit in FIG.~\ref{fig:AmyTD4-2d} comprises nine layers, of which four contain T gates, giving a T-depth of 4 and a circuit depth of 9. Likewise, in FIG.~\ref{fig:AmyTD3-2d}, the circuit consists of eleven layers, with three containing T gates, yielding a T-depth of 3 and a circuit depth of 11.

\subsection{Basic Toffoli decompositions with ancilla, utilizing seven T gates on diverse 2D topology}
\label{sub:ancilla}
In \cite{amy2013ieee}, the authors also presented a basic Toffoli decomposition using seven T gates with T-depth 2, requiring only one ancilla qubit. In this logical circuit (FIG.~\ref{fig:AmyTD2}), both control qubits connect to the other three qubits, implying degree-3 connectivity in a 2D layout if implemented without SWAP overhead. The target and ancilla qubits connect only to the control qubits, so degree-2 connectivity suffices for them. Consequently, the decomposition can be mapped to a square-grid architecture with diagonal connectivity, as shown in FIG.~\ref{fig:AmyTD2-2d}. The circuit contains 13 layers, of which 2 contain T gates, giving a circuit depth of 13 and T-depth of 2.
\begin{figure}[ht]
	\centering
	\begin{subfigure}{0.48\textwidth}
		\centering
		\includegraphics[width=1\linewidth]{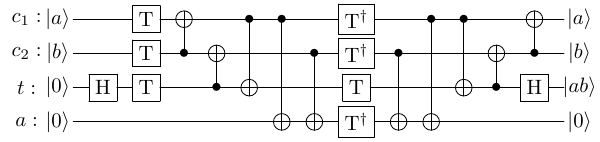}
		\caption{}
		\label{fig:AmyTD2}
	\end{subfigure}\\
	\begin{subfigure}{0.48\textwidth}
		\centering
		\scalebox{0.84}{\input{AmyTD2-2d}}
		\caption{}
		\label{fig:AmyTD2-2d}
	\end{subfigure}
	\caption{(a) T-depth-2 Toffoli decomposition from~\cite{amy2013ieee}, and (b) its corresponding 2D implementation.}
	\label{fig:td2}
\end{figure}

Later in the same year, Selinger~\cite{selinger2013pra} proposed a T-depth-1 Toffoli decomposition using seven T gates and four ancilla qubits (FIG.~\ref{fig:SelTD1}). The required CNOT connectivity among the control qubits ($c_1,c_2$), target ($t$), and ancilla ($a_1,\dots,a_4$) is:
$$\displaylines{c_1: a_1,a_2\quad c_2: a_2,a_3\quad t: a_3,a_4\cr
	a_1: c_1,a_3,a_4\quad a_2: c_1,c_2\quad a_3: c_2,t,a_1\quad a_4: t, a_1.}$$
To achieve this connectivity in a planar architecture, we use a 2D layout formed by combining two square-grid layouts with a triangular layout (FIG.~\ref{fig:SelTD1-2d}). The seven-layer decomposition contains T gates in exactly one layer, yielding a circuit depth of 7 and T-depth of 1.
\begin{figure}[htbp]
	\centering
	\begin{subfigure}{0.48\textwidth}
		\centering
		\includegraphics[width=1\linewidth]{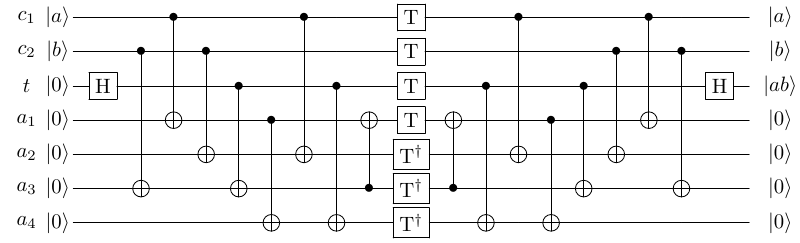}
		\caption{}
		\label{fig:SelTD1}
	\end{subfigure}\\
	\begin{subfigure}{0.48\textwidth}
		\centering
		\scalebox{0.9}{\input{SelTD1-2d}}
		\caption{}
		\label{fig:SelTD1-2d}
	\end{subfigure}
	\caption{(a) T-depth-1 Toffoli decomposition from~\cite{selinger2013pra}, and (b) its corresponding 2D implementation.}
	\label{fig:td1}
\end{figure}

\subsection{Gidney's logical-AND decomposition on linear nearest-neighbor topology}
\label{sub:logicalAND}
In Gidney's logical-AND decomposition \cite{gidney2018quantum} (see FIG.~\ref{fig:logicalAND-wm}), no CNOT gate is applied between the two control qubits; instead, each CNOT involves one of the control qubits and the target. This structure naturally aligns with a linear nearest-neighbor layout, keeping the target qubit in between. This placement ensures that the CNOT gates can be executed without requiring any SWAP gates. Under this arrangement, the logical-AND-based Toffoli decomposition attains a T-depth of 2: one T gate in layer 2, and the other three applied simultaneously on all three qubits in layer 7. The resulting 2D layout implementation of the logical-AND decomposition is illustrated in FIG.~\ref{fig:logicalAND-2d}. Note that, for this 2D realization, we consider only the computation phase of the Toffoli and omit the measurement-based uncomputation.
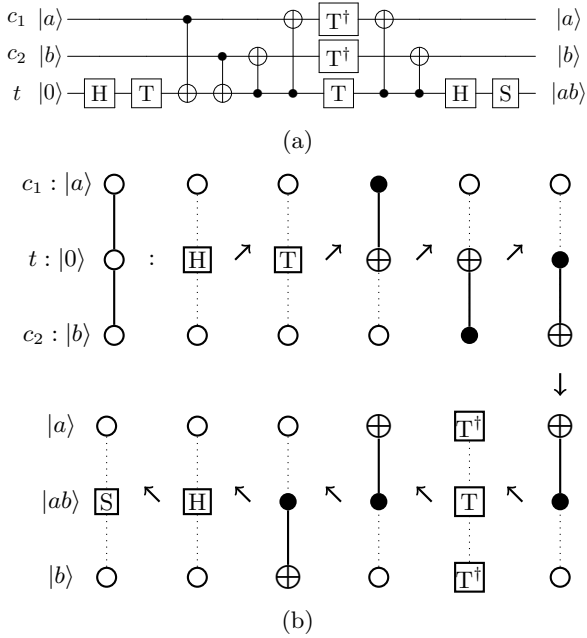
\begin{figure}[htbp]
	\centering
	\begin{subfigure}{0.48\textwidth}
		\centering
		\includegraphics[width=0.9\linewidth]{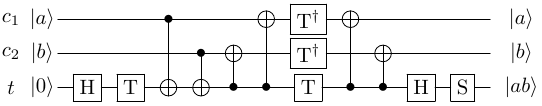}
		\caption{}
		\label{fig:logicalAND-wm}
	\end{subfigure}\\
	\begin{subfigure}{0.48\textwidth}
		\centering
		\scalebox{1}{\begin{tikzpicture}[
  v/.style={circle, draw=black, thick, minimum size=2.5mm, inner sep=1.6pt},
  v_sq/.style={rectangle, draw=black, thick, minimum size=3mm, inner sep=1.6pt},
  v_ctr/.style={circle, draw=black, thick, fill=black, minimum size=2mm, inner sep=1.6pt},
  v_cn/.style={inner sep=0pt},
  v_nt/.style={fill = white, thick, minimum size=2.5mm, inner sep=0pt},
  gate_sq/.style={rectangle, draw=black, fill = white, thick, minimum size=3.8mm, inner sep=0pt},
  e/.style={black, thick},
  arr/.style={->, thick},
  lab/.style={font=\small}
] 

\node[] at (-0.85,1) {$c_1: \ket{a}$};
\node[] at (-0.85,0) {$t: \ket{0}$};
\node[] at (-0.85,-1) {$c_2: \ket{b}$};

\node[v] (a1) at (-0.1, 1) {};
\node[v] (a2) at (-0.1, 0) {};
\node[v] (a3) at (-0.1,-1) {};
\draw[e] (a1)--(a2) (a2)--(a3);
\node[] at (0.4, 0) {$\mathrm{:}$};

\node[v] (b1) at (1, 1) {};
\node[v_sq] (b2) at (1, 0) {$\mathrm{H}$};
\node[v] (b3) at (1, -1) {};
\draw[dotted] (b1) -- (b2);
\draw[dotted] (b3) -- (b2);
\draw[arr] (1.5,0) -- (1.7,0.2);

\node[v] (c1) at (2.2, 1) {};
\node[v_sq] (c2) at (2.2, 0) {$\mathrm{T}$};
\node[v] (c3) at (2.2, -1) {};
\draw[dotted] (c1) -- (c2);
\draw[dotted] (c3) -- (c2);
\draw[arr] (2.7,0) -- (2.9,0.2);

\node[v_ctr] (d1) at (3.4, 1) {};
\node[v_cn]  (d2) at (3.4, 0) {$\bigoplus$};
\node[v] (d3) at (3.4, -1) {};
\draw[e] (d1) -- (d2); 
\draw[dotted] (d3) -- (d2);
\draw[arr] (3.9,0) -- (4.1,0.2);

\node[v_cn] (d1) at (4.6, 0) {$\bigoplus$};
\node[v_ctr]  (d2) at (4.6, -1) {};
\node[v] (d3) at (4.6, 1) {};
\draw[e] (d1) -- (d2); 
\draw[dotted] (d3) -- (d2);
\draw[arr] (5.1,0) -- (5.3,0.2);

\node[v_ctr] (d1) at (5.8, 0) {};
\node[v_cn]  (d2) at (5.8,-1) {$\bigoplus$};
\node[v] (d3) at (5.8, 1) {};
\draw[e] (d1) -- (d2); 
\draw[dotted] (d3) -- (d2);

\draw[arr] (5.8,-1.5) -- (5.8,-1.8);

\node[v_cn] (d1) at (5.8, -2.2) {$\bigoplus$};
\node[v_ctr]  (d2) at (5.8, -3.2) {};
\node[v] (d3) at (5.8, -4.2) {};
\draw[e] (d1) -- (d2); 
\draw[dotted] (d3) -- (d2);

\draw[arr] (5.3,-3.2) -- (5.1,-3);

\node[gate_sq] (d1) at (4.6, -2.2) {$\mathrm{T}^\dagger$};
\node[gate_sq]  (d2) at (4.6, -3.2) {$\mathrm{T}$};
\node[gate_sq]  (d3) at (4.6, -4.2) {$\mathrm{T}^\dagger$};
\draw[dotted] (d1) -- (d2);
\draw[dotted] (d3) -- (d2);

\draw[arr] (4.1,-3.2) -- (3.9,-3);

\node[v_cn] (d1) at (3.4, -2.2) {$\bigoplus$};
\node[v_ctr]  (d2) at (3.4, -3.2) {};
\node[v] (d3) at (3.4, -4.2) {};
\draw[e] (d1) -- (d2); 
\draw[dotted] (d3) -- (d2);

\draw[arr] (2.9,-3.2) -- (2.7,-3);

\node[v_ctr] (d1) at (2.2, -3.2) {};
\node[v_cn]  (d2) at (2.2,-4.2) {$\bigoplus$};
\node[v] (d3) at (2.2, -2.2) {};
\draw[e] (d1) -- (d2); 
\draw[dotted] (d3) -- (d2);

\draw[arr] (1.7,-3.2) -- (1.5,-3);

\node[v_sq] (c2) at (1, -3.2) { $\mathrm{H}$ };
\node[v] (c1) at (1, -2.2) {};
\node[v] (c3) at (1, -4.2) {};
\draw[dotted] (c1) -- (c2);
\draw[dotted] (c3) -- (c2);

\draw[arr] (0.5,-3.2) -- (0.3,-3);

\node[v_sq] (c2) at (-0.2, -3.2) {$\mathrm{S}$};
\node[v] (c1) at (-0.2, -2.2) {};
\node[v] (c3) at (-0.2, -4.2) {};
\draw[dotted] (c1) -- (c2);
\draw[dotted] (c3) -- (c2);

\node[v_nt] (a1) at (-0.8, -2.2) {$\ket{a}$};
\node[v_nt] (a2) at (-0.8, -3.2) {$\ket{ab}$};
\node[v_nt] (a3) at (-0.8,-4.2) {$\ket{b}$};




\end{tikzpicture}}
		\caption{}
		\label{fig:logicalAND-2d}
	\end{subfigure}
	\caption{(a) Logical-AND based Toffoli decomposition from~\cite{gidney2018quantum} (without uncomputation), and (b) its corresponding 2D implementation.}
	\label{fig:logicalAND-htree}
\end{figure}

\subsection{T-depth-1 Toffoli decomposition due to Jaques et al. on square-lattice topology}
\label{sub:jaques}
In 2020, Jaques et al. \cite{jaques2020eurocrypt} proposed a measurement-based Toffoli decomposition (FIG.~\ref{fig:Jaques-wm}) achieving T-depth 1. The construction involves two control qubits $c_1,c_2$, a target qubit $t$, and an ancilla qubit $a$. Similar to the logical-AND circuit, and ignoring the uncomputation part, there is no CNOT interaction between the control qubits, nor between the target and the ancilla. This structural property allows us to place the control qubits $c_1$ and $c_2$ at diagonally opposite corners of the square-grid topology, while assigning the target $t$ and ancilla $a$ to the remaining two corners. The resulting circuit on the square-lattice architecture is shown in FIG.~\ref{fig:Jaques-2d}. Among the eight layers in the decomposition, only layer 4 contains T gates, yielding an overall circuit depth of 8 and a $\mathrm{T}$-depth of 1.
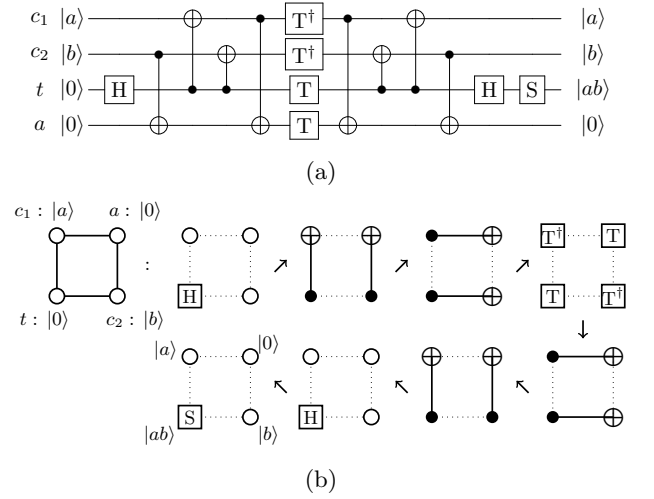
\begin{figure}[htbp]
	\centering
	\begin{subfigure}{0.48\textwidth}
		\centering
		\includegraphics[width=0.9\linewidth]{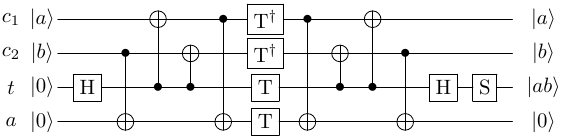}
		\caption{}
		\label{fig:Jaques-wm}
	\end{subfigure}\\
	\begin{subfigure}{0.48\textwidth}
		\centering
		\scalebox{0.8}{\begin{tikzpicture}[
  open/.style={circle, draw=black, fill = white, thick, minimum size=2.5mm, inner sep=0pt},
  open_nt/.style={fill = white, thick, minimum size=3mm, inner sep=0pt},
  open_cn/.style={circle, minimum size=2.5mm, inner sep=0pt},
  open_ctr/.style={circle, minimum size=2mm,fill = black, inner sep=0pt},
  gate_sq/.style={rectangle, draw=black, fill = white, thick, minimum size=3.8mm, inner sep=0pt},
  filled/.style={circle, draw=black, thick, fill=black, minimum size=5mm, inner sep=0pt},
  op/.style={circle, draw=black, thick, minimum size=6mm, inner sep=0pt, font=\small},
  edge/.style={black, thick},
  edge_d/.style={black, thick,dashed},
  gate/.style={draw=black, thick, minimum width=6mm, minimum height=6mm, inner sep=1pt, font=\small},
  arr/.style={->, thick}
]

\newcommand{\SquareGraph}[5]{%
  \begin{scope}[shift={(#1,0)}]
    \coordinate (c1) at (0,1);
    \coordinate (a) at (1,1);
    \coordinate (t) at (0,0);
    \coordinate (c2) at (1,0);

    \draw[edge] (c1)--(a);
    \draw[edge] (c1)--(t);
    \draw[edge] (t)--(c2);
    \draw[edge] (a)--(c2);

    \node[open] (nc1) at (c1) {#2};
    \node[open] (na) at (a) {#3};
    \node[open] (nt) at (t) {#4};
    \node[open] (nc2) at (c2) {#5};
  \end{scope}
}


\SquareGraph{0}{}{}{}{}
\node[] at (-0.2,1.4) {$c_1:\,\ket{a}$};
\node[] at (-0.2,-0.4) {$t:\,\ket{0}$};
\node[] at (1.3,1.4) {$a: \,\ket{0}$};
\node[] at (1.3,-0.4) {$c_2:\,\ket{b}$};
%
\node[] at (1.5, 0.5) {$\mathrm{:}$};
\begin{scope}[shift={(2.2,0)}]
    \coordinate (c1) at (0,1);
    \coordinate (a) at (1,1);
    \coordinate (t) at (0,0);
    \coordinate (c2) at (1,0);

    \draw[dotted] (c1)--(a);
    \draw[dotted] (c1)--(t);
    \draw[dotted] (t)--(c2);
    \draw[dotted] (a)--(c2);

    \node[open] (nc1) at (c1) { };
    \node[open] (na) at (a) { };
    \node[gate_sq] (nt) at (t) {$\mathrm{H}$};
    \node[open] (nc2) at (c2) { };

\end{scope}

\draw[arr] (3.6,0.4) -- (3.8,0.6);

\begin{scope}[shift={(4.2,0)}]
    \coordinate (c1) at (0,1);
    \coordinate (a) at (1,1);
    \coordinate (t) at (0,0);
    \coordinate (c2) at (1,0);

    \draw[dotted] (c1)--(a);
    \draw[edge] (c1)--(t);
    \draw[dotted] (t)--(c2);
    \draw[edge] (a)--(c2);

    \node[open_cn] (nc1) at (c1) {$\bigoplus$};
    \node[open_cn] (na) at (a) {$\bigoplus$};
    \node[open_ctr] (nt) at (t) {};
    \node[open_ctr] (nc2) at (c2) {};
\end{scope}

\draw[arr] (5.6,0.4) -- (5.8,0.6);

\begin{scope}[shift={(6.2,0)}]
    \coordinate (c1) at (0,1);
    \coordinate (a) at (1,1);
    \coordinate (t) at (0,0);
    \coordinate (c2) at (1,0);

    \draw[edge] (c1)--(a);
    \draw[dotted] (c1)--(t);
    \draw[edge] (t)--(c2);
    \draw[dotted] (a)--(c2);

    \node[open_ctr] (nc1) at (c1) {};
    \node[open_cn] (na) at (a) {$\bigoplus$};
    \node[open_ctr] (nt) at (t) {};
    \node[open_cn] (nc2) at (c2) {$\bigoplus$};
\end{scope}

\draw[arr] (7.6,0.4) -- (7.8,0.6);

\begin{scope}[shift={(8.2,0)}]
    \coordinate (c1) at (0,1);
    \coordinate (a) at (1,1);
    \coordinate (t) at (0,0);
    \coordinate (c2) at (1,0);

    \draw[dotted] (c1)--(a);
    \draw[dotted] (c1)--(t);
    \draw[dotted] (t)--(c2);
    \draw[dotted] (a)--(c2);

    \node[gate_sq] (nc1) at (c1) {$\mathrm{T}^{\dagger}$};
    \node[gate_sq] (na) at (a) {$\mathrm{T}$};
    \node[gate_sq] (nt) at (t) {$\mathrm{T}$};
    \node[gate_sq] (nc2) at (c2) {$\mathrm{T}^{\dagger}$};

\end{scope}
\draw[arr] (8.7,-0.4) -- (8.7,-0.7) ;


\begin{scope}[shift={(8.2,-2)}]
    \coordinate (c1) at (0,1);
    \coordinate (a) at (1,1);
    \coordinate (t) at (0,0);
    \coordinate (c2) at (1,0);

    \draw[edge] (c1)--(a);
    \draw[dotted] (c1)--(t);
    \draw[edge] (t)--(c2);
    \draw[dotted] (a)--(c2);

    \node[open_ctr] (nc1) at (c1) {};
    \node[open_cn] (na) at (a) {$\bigoplus$};
    \node[open_ctr] (nt) at (t) {};
    \node[open_cn] (nc2) at (c2) {$\bigoplus$};
\end{scope}

\draw[arr] (7.8,-1.6) -- (7.6,-1.4);

\begin{scope}[shift={(6.2,-2)}]

    \coordinate (c1) at (0,1);
    \coordinate (a) at (1,1);
    \coordinate (t) at (0,0);
    \coordinate (c2) at (1,0);

    \draw[dotted] (c1)--(a);
    \draw[edge] (c1)--(t);
    \draw[dotted] (t)--(c2);
    \draw[edge] (a)--(c2);

    \node[open_cn] (nc1) at (c1) {$\bigoplus$};
    \node[open_cn] (na) at (a) {$\bigoplus$};
    \node[open_ctr] (nt) at (t) {};
    \node[open_ctr] (nc2) at (c2) {};
\end{scope}

\draw[arr] (5.8,-1.6) -- (5.6,-1.4);


\begin{scope}[shift={(4.2,-2)}]
    \coordinate (c1) at (0,1);
    \coordinate (a) at (1,1);
    \coordinate (t) at (0,0);
    \coordinate (c2) at (1,0);

    \draw[dotted] (c1)--(a);
    \draw[dotted] (c1)--(t);
    \draw[dotted] (t)--(c2);
    \draw[dotted] (a)--(c2);

    \node[open] (nc1) at (c1) {};
    \node[open] (na) at (a) {};
    \node[gate_sq] (nt) at (t) {$\mathrm{H}$};
    \node[open] (nc2) at (c2) {};

\end{scope}

\draw[arr] (3.8,-1.6) -- (3.6,-1.4);


\begin{scope}[shift={(2.2,-2)}]
    \coordinate (c1) at (0,1);
    \coordinate (a) at (1,1);
    \coordinate (t) at (0,0);
    \coordinate (c2) at (1,0);

    \draw[dotted] (c1)--(a);
    \draw[dotted] (c1)--(t);
    \draw[dotted] (t)--(c2);
    \draw[dotted] (a)--(c2);

    \node[open] (nc1) at (c1) {};
    \node[open] (na) at (a) {};
    \node[gate_sq] (nt) at (t) {$\mathrm{S}$};
    \node[open] (nc2) at (c2) {};
\end{scope}

\node[] at (1.8,-0.9) {$\ket{a}$};
\node[] at (1.7,-2.3) {$\ket{ab}$};
\node[] at (3.5,-0.8) {$\ket{0}$};
\node[] at (3.5,-2.3) {$\ket{b}$};





\end{tikzpicture}}
		\caption{}
		\label{fig:Jaques-2d}
	\end{subfigure}
	\caption{(a) T-depth-1 Toffoli decomposition from~\cite{jaques2020eurocrypt} (without uncomputation), and (b) its corresponding 2D implementation.}
	\label{fig:jaques}
\end{figure}

\subsection{3-MCT decomposition using 6 T gates on IBM's 5-qubit 2D topologies}
In 2021, Gidney et al.~\cite{gidney2021cccz} presented the first CCCZ circuit requiring only six T gates. As is well known, an $n$-controlled Z gate $\mathrm{C}^n\mathrm{Z}$ gate, combined with Hadamard gates on the target qubit before and after the operation, yields an $n$-controlled Toffoli gate $\mathrm{C}^n\mathrm{X}$~\cite[FIG.~1]{dutta2025pra}. The circuit consists of three control qubits $c_1, c_2, c_3$, a target qubit $t$, and an ancilla $a$. From the logical circuit (FIG.~\ref{fig:cccz}), all CNOT interactions involve the ancilla; no other pair of qubits interacts via CNOT. Consequently, the ancilla requires degree-4 connectivity to avoid SWAP overhead. Furthermore, the post-measurement three-qubit controlled-Z operations act on either ${c_1,c_2, a}$ or ${c_3,t, a}$. This implies that $c_1$ and $c_2$ must be adjacent, each requiring degree-2 connectivity, and likewise $c_3$ and $t$ must be adjacent with 2-degree connectivity. Satisfying all these adjacency constraints is achievable on IBM’s five-qubit layout, shown in FIG.~\ref{fig:cccz-2d}. The corresponding 2D implementation achieves a circuit depth of 18 and a T-depth of 6.
\begin{figure*}[htbp]
	\centering
	\begin{subfigure}{1\textwidth}
		\centering
		\includegraphics[width=0.78\linewidth]{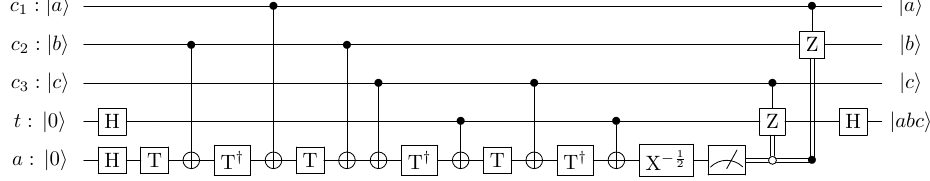}
		\caption{}
		\label{fig:cccz}
	\end{subfigure}
	\begin{subfigure}{1\textwidth}
		\centering
		\scalebox{0.95}{\input{cccz-2d}}
		\caption{}
		\label{fig:cccz-2d}
	\end{subfigure}
	\caption{(a) 6 T-depth 3-MCT decomposition from~\cite{gidney2021cccz}, and (b) its corresponding 2D implementation.}
	\label{fig:cccz5q}
\end{figure*}

Later, Nakanishi et al. proposed a modified CCCZ construction (FIG.~\ref{fig:naka}) with the same interaction pattern but reduced T-depth of 2. Its 2D implementation without SWAPs also fits naturally on the same IBM five-qubit layout, as shown in FIG.~\ref{fig:naka-2d}. The resulting circuit has an overall depth of 19 and a T-depth of 2.
\begin{figure*}[htbp]
	\centering
	\begin{subfigure}{1\textwidth}
		\centering
		\includegraphics[width=0.72\linewidth]{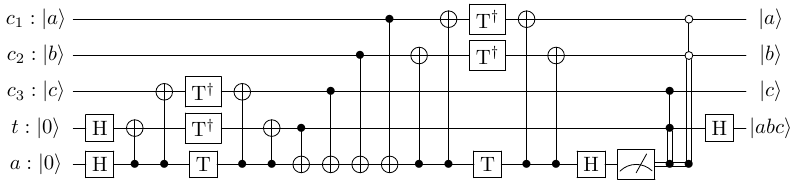}
		\caption{}
		\label{fig:naka}
	\end{subfigure}
	\begin{subfigure}{1\textwidth}
		\centering
		\scalebox{0.93}{\input{naka-2D}}
		\caption{}
		\label{fig:naka-2d}
	\end{subfigure}
	\caption{(a) 2 T-depth 3-MCT decomposition from~\cite{nakanishi2024arxiv}, and (b) its corresponding 2D implementation.}
	\label{fig:naka5q}
\end{figure*}

A comprehensive summary of basic Toffoli decompositions transformed onto various 2D qubit layouts, together with their resource requirements under different metrics, is provided in TABLE~\ref{tab:tof-2d}.
\begin{table*}[htbp]
	\centering
	\begin{tabular}{c c c c c c}
		\hline
		\hline
		~Topology~ & ~Ancilla~ & ~T-Count~ & ~T-Depth~ & ~Circuit depth~ & ~Logical circuit~ \\[0.1cm]
		\hline
		~Triangular (FIG.~\ref{fig:NCTD6-2d})~ & {\color{red}0} & 7 & 6 & 13 & ~FIG.~\ref{fig:NCTD6}~\\
		~Triangular (FIG.~\ref{fig:AmyTD4-2d})~ & {\color{red}0} & 7 & 4 & 9 & ~FIG.~\ref{fig:AmyTD4}~\\
		~Triangular (FIG.~\ref{fig:AmyTD3-2d})~ & {\color{red}0} & 7 & 3 & 11 & ~FIG.~\ref{fig:AmyTD3}~\\
		~Square with a diagonal (FIG.~\ref{fig:AmyTD2-2d})~ & 1 & 7 & 2 & 13 & ~FIG.~\ref{fig:AmyTD2}~\\
		~Square + Triangular (FIG.~\ref{fig:SelTD1-2d})~ & 4 & 7 & {\color{teal}1} & 7 & ~FIG.~\ref{fig:SelTD1}~\\
		~Linear (FIG.~\ref{fig:logicalAND-2d})~ & {\color{red}0} & {\color{blue}4} & $1 + 1$ & 11 & ~FIG.~\ref{fig:logicalAND-wm}~\\
		~Square (FIG.~\ref{fig:Jaques-2d})~ & 1 & {\color{blue}4} & {\color{teal}1} & 8 & ~FIG.~\ref{fig:Jaques-wm}~\\[2pt]
		~IBM's 5-qubit (FIG.~\ref{fig:cccz-2d})~ & 1 & 6 & 6 & 18 & ~FIG.~\ref{fig:cccz}~\\
		~IBM's 5-qubit (FIG.~\ref{fig:naka-2d})~ & 1 & 6 & 2 & 20 & ~FIG.~\ref{fig:naka}~\\
		\hline
	\end{tabular}
	\caption{Mapping state-of-the-art basic Toffoli decompositions in different 2D qubit layouts. The circuit with the lowest T-count is highlighted in {\color{blue}blue}, the one with the least T-depth in {\color{teal}teal}, and the ancilla-free decompositions in {\color{red}red}. The last two rows correspond to the 3-MCT decompositions in the 5-qubit 2D layout provided by IBM.}
	\label{tab:tof-2d}
\end{table*}

\section{Optimal Toffoli-depth MCT decompositions on infinite grid 2D layouts without SWAP}
\label{sec:cont1}
In this section, we explore optimal Toffoli-depth MCT decompositions~\cite{dutta2025pra} tailored to various 2D qubit layouts. Recall from~\cite{dutta2025pra} that an optimal Toffoli-depth $n$-MCT decomposition requires $n-2$ ancilla qubits and $n-1$ Toffoli gates, achieving a Toffoli depth of $\lceil \log_2 n \rceil$. The overall T-count and T-depth depend on the specific basic Toffoli decomposition employed from Section.~\ref{sec:warmup}.

\subsection{MCT decompositions on triangular topologies}
\label{sub:mct-triangular}
The optimal Toffoli-depth $n$-MCT decompositions on an infinite grid triangular 2D layout for $n = 3,4,5,6,7,8$ are shown in FIG.~\ref{fig:mct-triangle}.
\begin{figure}[htbp]
	\centering
	\begin{subfigure}{0.12\textwidth}
		\centering
		\includegraphics[width=1\linewidth]{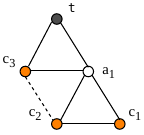}
		\caption{}
		\label{fig:t3}
	\end{subfigure}
	\begin{subfigure}{0.13\textwidth}
		\centering
		\includegraphics[width=1\linewidth]{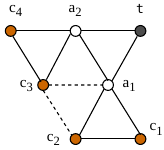}
		\caption{}
		\label{fig:t4}
	\end{subfigure}
	\begin{subfigure}{0.17\textwidth}
		\centering
		\includegraphics[width=1\linewidth]{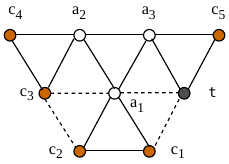}
		\caption{}
		\label{fig:t5}
	\end{subfigure}\\
	\begin{subfigure}{0.22\textwidth}
		\centering
		\includegraphics[width=1\linewidth]{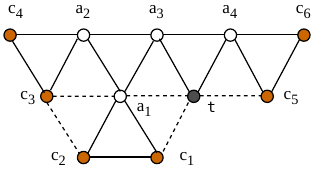}
		\caption{}
		\label{fig:t6}
	\end{subfigure}
	\begin{subfigure}{0.24\textwidth}
		\centering
		\includegraphics[width=1\linewidth]{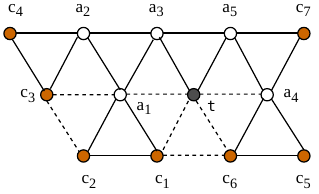}
		\caption{}
		\label{fig:t7}
	\end{subfigure}\\
	\begin{subfigure}{0.3\textwidth}
		\centering
		\includegraphics[width=1\linewidth]{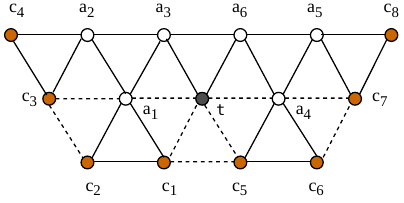}
		\caption{}
		\label{fig:t8}
	\end{subfigure}
	\caption{$n$-MCT circuit decompositions on triangular 2D qubit layouts for $n = 3, 4, 5, 6, 7, 8$. The Brown circles denote control qubits, white circles denote ancilla for MCT decompositions, and the gray circle denotes the target qubit.} 
	\label{fig:mct-triangle}
\end{figure}

In FIG.~\ref{fig:mct-triangle}, the number of solid-edged triangles corresponds to the number of Toffoli gates. Since the T-depth-3 construction of Amy et al.~(FIG.~\ref{fig:AmyTD3-2d}) fits the triangular topology without any SWAP overhead, further decomposition of each Toffoli gate using this construction yields an overall T-count seven times the Toffoli count and a T-depth of $3\lceil \log_2 n \rceil$. Although the construction assumes an infinite grid, 2D triangular topology, the layouts in FIG.~\ref{fig:mct-triangle} preserve the qubit count of the corresponding logical circuits: $n+1$ working qubits and $n-2$ ancilla, totaling $2n-1$ qubits. Complete resource estimates for different values of $n$ are given in TABLE~\ref{tab:mct-triangle}.
\begin{table}[htbp]
	\centering
	\setlength{\tabcolsep}{0pt}
	\begin{tabular}{c c c c c c}
		\hline
		\hline
		~$n$~ & ~\#Qubit~ & \#Toffoli~ & Toffoli depth & ~T-Count & ~T-Depth~ \\[0.1cm]
		\hline
		~3~ & 5 & 2 & 2 & 14 & 6\\
		~4~ & 7 & 3 & 2 & 21 & 6\\
		~5~ & 9 & 4 & 3 & 28 & 9\\
		~6~ & 11 & 5 & 3 & 35 & 9\\
		~7~ & 13 & 6 & 3 & 42 & 9\\
		~8~ & 15 & 7 & 3 & 49 & 9\\
		~$n$~ & $2n-1$ & $n-1$ & $\lceil \log_2 n\rceil$ & $7(n-1)$ & $3\lceil \log_2 n\rceil$\\
		\hline
	\end{tabular}
	\caption{Quantum resource requirements for $n$-MCT decompositions on triangular 2D layouts, using the T-depth-3 basic Toffoli decomposition shown in FIG.~\ref{fig:AmyTD3-2d}.}
	\label{tab:mct-triangle}
\end{table}

Depending on design requirements, other SWAP-free Toffoli decompositions compatible with the triangular 2D topology, such as those in FIG.~\ref{fig:NCTD6-2d} and FIG.~\ref{fig:AmyTD4-2d}, may also be used for the final decomposition.

\subsection{MCT decomposition on square-grid topologies}
\label{sub:mct-square}
The optimal Toffoli-depth $n$-MCT decompositions on an infinite 2D square-grid architecture for $n \in\{3,4,5,6,7,8\}$ are shown in FIG.~\ref{fig:mct-square}.
\begin{figure}[!htbp]
	\centering
	\begin{subfigure}{0.14\textwidth}
		\centering
		\includegraphics[width=1\linewidth]{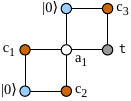}
		\caption{}
		\label{fig:s3}
	\end{subfigure}\quad\quad\quad
	\begin{subfigure}{0.18\textwidth}
		\centering
		\includegraphics[width=1\linewidth]{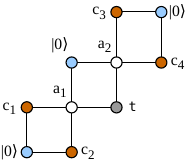}
		\caption{}
		\label{fig:s4}
	\end{subfigure}\\[5pt]
	\begin{subfigure}{0.16\textwidth}
		\centering
		\includegraphics[width=1\linewidth]{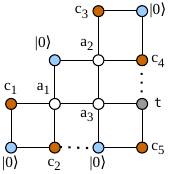}
		\caption{}
		\label{fig:s5}
	\end{subfigure}\quad\,\,
	\begin{subfigure}{0.21\textwidth}
		\centering
		\includegraphics[width=1\linewidth]{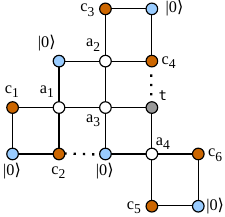}
		\caption{}
		\label{fig:s6}
	\end{subfigure}\\[4pt]
	\begin{subfigure}{0.25\textwidth}
		\centering
		\includegraphics[width=1\linewidth]{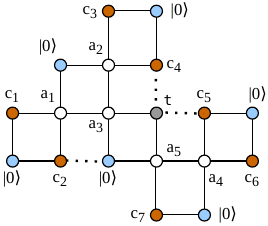}
		\caption{}
		\label{fig:s7}
	\end{subfigure}\\[5pt]
	\begin{subfigure}{0.26\textwidth}
		\centering
		\includegraphics[width=1\linewidth]{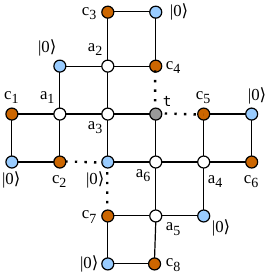}
		\caption{}
		\label{fig:s8}
	\end{subfigure}
	\caption{$n$-MCT circuit decomposition on square-grid 2D layouts for $n = 3, 4, 5, 6, 7, 8$. Here, the brown circles denote control qubits, white circles denote ancilla for MCT decompositions, blue circles denotes ancilla used for basic Toffoli decompositions, and the gray circle denotes the target qubit.}
	\label{fig:mct-square}
\end{figure}

In FIG.~\ref{fig:mct-square}, the number of solid-edged squares corresponds to the number of Toffoli gates. Because the T-depth-1 Toffoli decomposition of Jaques et al.~\cite{jaques2020eurocrypt} (see FIG.~\ref{fig:Jaques-2d}) fits within the square-grid architecture without requiring any SWAP operations, the overall T-depth also remains $\lceil \log_2 n \rceil$, and the resulting T-count becomes four times the Toffoli count. As in the triangular infinite-grid implementations, we constrain the available qubits to match the logical circuits: $n+1$ working qubits, $n-2$ ancilla for the binary-tree-based decomposition, and $n-1$ ancilla for the T-depth-1 variant, for a total of $3n-2$ qubits. Complete resource estimates for the decomposition of $n$-MCT gate on an infinite-grid triangular topology for various values of $n$ are presented in TABLE~\ref{tab:mct-square}.
\begin{table}[htbp]
	\centering
	\setlength{\tabcolsep}{0pt}
	\begin{tabular}{c c c c c c}
		\hline
		\hline
		~$n$~ & ~\#Qubit~ & \#Toffoli~ & Toffoli depth & ~T-Count & ~T-Depth~ \\[0.1cm]
		\hline
		~3~ & 7 & 2 & 2 & 8 & 2\\
		~4~ & 10 & 3 & 2 & 12 & 2\\
		~5~ & 13 & 4 & 3 & 16 & 3\\
		~6~ & 16 & 5 & 3 & 20 & 3\\
		~7~ & 19 & 6 & 3 & 24 & 3\\
		~8~ & 12 & 7 & 3 & 28 & 3\\
		~$n$~ & $3n-2$ & $n-1$ & $\lceil \log_2 n\rceil$ & $4(n-1)$ & $\lceil \log_2 n\rceil$\\
		\hline
	\end{tabular}
	\caption{Quantum resource requirements for $n$-MCT decompositions on 2D square-grid layouts, using the T-depth-1 basic Toffoli decomposition shown in FIG.~\ref{fig:Jaques-2d}.}
	\label{tab:mct-square}
\end{table}

Additionally, for basic Toffoli decompositions, one may also consider the T-depth-2 construction by Amy et al.~\cite{amy2013ieee}, which can be implemented on a square-grid layout (without SWAP) with added diagonal connectivity, as shown in FIG.~\ref{fig:AmyTD2-2d}. The corresponding 8-MCT decomposition is presented in FIG.~\ref{fig:mct-sqd}, and the associated quantum resource estimates are summarized in TABLE~\ref{tab:mct-sqd}.
\begin{figure}[htbp]
	\centering
	\includegraphics[width=0.54\linewidth]{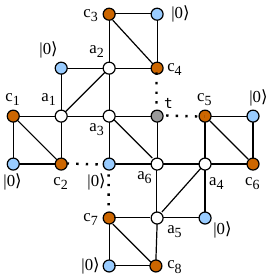}
	\caption{$8$-MCT circuit decomposition on 2D square grid with a diagonal connectivity. Brown circles denote control qubits, white circles denote ancilla used for MCT decompositions, blue circles denotes ancilla used for basic Toffoli decompositions, and the gray circle denotes the target qubit.}
	\label{fig:mct-sqd}
\end{figure}
\begin{table}[htbp]
	\centering
	\setlength{\tabcolsep}{0pt}
	\begin{tabular}{c c c c c c}
		\hline
		\hline
		~$n$~ & ~\#Qubit~ & \#Toffoli~ & Toffoli depth & ~T-Count & ~T-Depth~ \\[0.1cm]
		\hline
		~3~ & 7 & 2 & 2 & 14 & 4\\
		~4~ & 10 & 3 & 2 & 21 & 4\\
		~5~ & 13 & 4 & 3 & 28 & 6\\
		~6~ & 16 & 5 & 3 & 35 & 6\\
		~7~ & 19 & 6 & 3 & 42 & 6\\
		~8~ & 12 & 7 & 3 & 49 & 6\\
		~$n$~ & $3n-2$ & $n-1$ & $\lceil \log_2 n\rceil$ & $7(n-1)$ & $2\lceil \log_2 n\rceil$\\
		\hline
	\end{tabular}
	\caption{Quantum resource requirements for $n$-MCT decompositions on 2D square-grid layouts, using the T-depth-2 basic Toffoli decomposition shown in FIG.~\ref{fig:AmyTD2-2d}.}
	\label{tab:mct-sqd}
\end{table}

\subsection{MCT decompositions on H-tree topologies}
\label{sub:mct-htree}
In \cite{hu2019dac}, the authors presented $n$-MCT decompositions on an H-tree layout using relative-phase Toffoli gates. Here, we provide $n$-MCT decompositions on an infinite-grid H-tree layout for $n=3,4,5,6,7,8$, as shown in FIG.~\ref{fig:mct-htree}.
\begin{figure*}[htbp]
	\centering
	\begin{subfigure}{0.12\textwidth}
		\centering
		\includegraphics[width=0.9\linewidth]{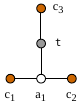}
		\caption{}
		\label{fig:h3}
	\end{subfigure}\quad\quad\quad
	\begin{subfigure}{0.12\textwidth}
		\centering
		\includegraphics[width=0.9\linewidth]{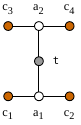}
		\caption{}
		\label{fig:h4}
	\end{subfigure}\quad\quad\quad
	\begin{subfigure}{0.22\textwidth}
		\centering
		\includegraphics[width=0.9\linewidth]{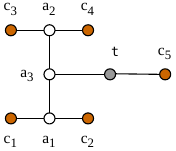}
		\caption{}
		\label{fig:h5}
	\end{subfigure}\quad\quad
	\begin{subfigure}{0.25\textwidth}
		\centering
		\includegraphics[width=0.9\linewidth]{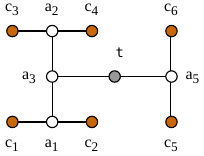}
		\caption{}
		\label{fig:h6}
	\end{subfigure}\\[8pt]
	\begin{subfigure}{0.32\textwidth}
		\centering
		\includegraphics[width=0.9\linewidth]{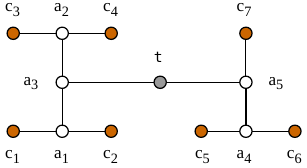}
		\caption{}
		\label{fig:h7}
	\end{subfigure}\quad\quad\quad\quad\quad
	\begin{subfigure}{0.34\textwidth}
		\centering
		\includegraphics[width=0.9\linewidth]{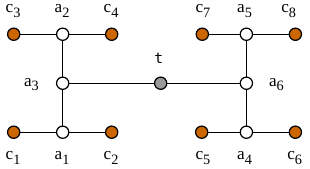}
		\caption{}
		\label{fig:h8}
	\end{subfigure}
	\caption{$n$-MCT circuit decompositions on H-tree layouts for $n = 3, 4, 5, 6, 7, 8$. Brown circles denote control qubits, white circles denote ancilla qubits, and the grey circle denotes the target qubit.}
	\label{fig:mct-htree}
\end{figure*}

As the Gidney’s logical-AND construction \cite{gidney2018quantum} naturally aligns with the H-tree layout, we use this for further decomposition of the basic Toffoli gates (see FIG.~\ref{fig:logicalAND-2d}). This yields an overall T-depth of $\lceil\log_2 n\rceil +1$, which is lower than that in \cite{hu2019dac}. In FIG.~\ref{fig:mct-htree}, the number of straight lines corresponds to the Toffoli count, and therefore the T-count is four times the Toffoli count. The qubit requirements match the logical circuit, consisting of $n+1$ working qubits and $n-2$ ancilla for the optimal-depth decomposition. TABLE~\ref{tab:mct-htree} summarizes the quantum resource requirements for implementing $n$-MCT gates on the H-tree topology using the logical-AND based decomposition for $n=3,4,5,6,7,8$.
\begin{table}[htbp]
	\centering
	\setlength{\tabcolsep}{0pt}
	\begin{tabular}{c c c c c c}
		\hline
		\hline
		~$n$~ & ~\#Qubit~ & \#Toffoli~ & Toffoli depth & ~T-Count & ~T-Depth~ \\[0.1cm]
		\hline
		~3~ & 5 & 2 & 2 & 8 & 3\\
		~4~ & 7 & 3 & 2 & 12 & 3\\
		~5~ & 9 & 4 & 3 & 16 & 4\\
		~6~ & 11 & 5 & 3 & 20 & 4\\
		~7~ & 13 & 6 & 3 & 24 & 4\\
		~8~ & 15 & 7 & 3 & 28 & 4\\
		~$n$~ & $2n-1$ & $n-1$ & $\lceil \log_2 n\rceil$ & $4(n-1)$ & $\lceil \log_2 n\rceil +1$\\
		\hline
	\end{tabular}
	\caption{Quantum resource requirements for MCT decompositions on 2D H-tree layouts, using the logical-AND based basic Toffoli decomposition shown in FIG.~\ref{fig:logicalAND-2d}.}
	\label{tab:mct-htree}
\end{table}

\subsection{Comparison of quantum resource-requirements for MCT decomposition in different 2D layouts}
\label{sub:mct-comparison}
Here, we compare the quantum resource requirements for MCT decompositions across different qubit layouts, including triangular, square-grid, square-grid with a diagonal, and H-tree networks. Earlier, we discussed the resource requirements for MCT-to-Toffoli and Toffoli-to-T decompositions in different 2D topologies and summarized estimates for both specific instances and the general $n$-qubit case. TABLE \ref{tab:mct-comparison} now provides a comparison of these resource costs across different 2D topologies in terms of generic $n$.
\begin{table*}[htbp]
	\centering
	\begin{tabular}{c c c c c c}
		\hline
		\hline
		~Topology~ & ~\#Qubit~ & ~\#Toffoli~ & ~Toffoli depth~ & ~T-Count~ & ~T-Depth~\\[0.1cm]
		\hline
		Triangular & $2n-1$ & $n-1$ & $\lceil \log_2 n\rceil$ & $7(n-1)$ & $3\lceil \log_2 n\rceil$\\
		Square-grid & $3n-2$ & $n-1$ & $\lceil \log_2 n\rceil$ & $4(n-1)$ & $\lceil \log_2 n\rceil$\\
		~Square-grid with a diagonal~ & $3n-2$ & $n-1$ & $\lceil \log_2 n\rceil$ & $7(n-1)$ & $2\lceil \log_2 n\rceil$\\
		H-tree & $2n-1$ & $n-1$ & $\lceil \log_2 n\rceil$ & $4(n-1)$ & ~$\lceil \log_2 n\rceil +1$~\\
		\hline
	\end{tabular}
	\caption{Quantum resource requirements for $n$-MCT decompositions on different 2D layouts.}
	\label{tab:mct-comparison}
\end{table*}

From TABLE \ref{tab:mct-comparison}, it is evident that under an infinite-grid assumption, the minimum T-depth is achieved in square-grid layouts using the basic Toffoli decomposition of Jaques et al.~\cite{jaques2020eurocrypt}. In this setting, the total number of required qubits, including working and ancilla qubits, is $3n - 2$. However, when mapping the logical circuits to an H-tree layout using the logical-AND decomposition of Gidney~\cite{gidney2018quantum}, the minimum achievable T-depth increases by 1, while the qubit count decreases significantly to $2n - 1$. Note that a similar observation also holds for MCT decompositions in logical circuits with $(n-1)$-degree connectivity among qubits. Consequently, in the following section, where we assume restricted 2D qubit layouts, we focus only on those architectures in which the basic 4-T Toffoli decompositions from \cite{gidney2018quantum, jaques2020eurocrypt} can be implemented without SWAP.

\section{Bounding depth overhead in constrained topologies}
\label{sec:depth_bound}

In the previous section, we presented MCT decompositions assuming architectures with unrestricted topology size. In practice, however, the qubit layouts are fixed in size, and mapping an MCT decomposition onto such an architecture requires routing qubits to establish the necessary neighborhood interactions. This routing is typically achieved by inserting SWAP gates, which introduce additional circuit depth. In this section, we explore efficient qubit-placement strategies for the binary-tree-based $n$-MCT decomposition proposed in \cite{dutta2025pra} across different qubit layouts.

Routing in a quantum architecture aims to bring interacting qubits together while maximizing parallelism, thereby minimizing the physical depth corresponding to a given logical circuit. Mapping an $n$-MCT gate involves decomposing it into multiple Toffoli gates and subsequently placing and routing these gates within the target architecture. Consequently, the overall circuit depth depends critically on the chosen placement and routing strategy.

We further study packing strategies for embedding Toffoli gates into a given topology. To this end, we develop a refined bounding framework based on geometrically embedding the unit graphs introduced in Section~\ref{sec:cont1} into standard topologies via disjoint packing.

Before proceeding further, we fix the notation for the parameters under consideration. We consider an $n$-MCT decomposition on a topology $\mathcal{T}$. The total number of Toffoli gates in the decomposition is denoted by $T_g(n)$, while the number of Toffoli gates in logical layer $\ell$ is denoted by $g_\ell$. Further, $D_{\mathrm{logical}}$ denotes the Toffoli depth of the corresponding logical decomposition, which is $\lceil \log_2 n\rceil$ for an optimal depth decomposition following \cite{dutta2025pra}.

Additionally, the decomposition induces an interaction graph $M$, referred to as a motif, arising from the Toffoli-to-T decomposition structures (see FIG.~\ref{fig:logicalAND-2d} and FIG.~\ref{fig:Jaques-2d}). We denote by $P_M(\mathcal{T})$ the maximum number of vertex-disjoint embeddings of $M$ in $\mathcal{T}$. Finally, $D_{\pi}(\mathcal{T})$ denotes the upper bound on the permutation-routing depth of the topology $\mathcal{T}$. 

In this regard, we proceed with the following lemma.

\begin{lemma}
	\label{lem:lemma1}
	Let $D_M(\mathcal{T})$ denote the depth overhead for mapping an $n$-MCT decomposition onto a topology $\mathcal{T}$. Then
	$$D_M(\mathcal{T}) < D_{\pi}(\mathcal{T})
	\left[
	\frac{T_g(n)}{P_M(\mathcal{T})} +
	D_{\mathrm{logical}}
	\right].$$
\end{lemma}

\begin{proof}
	Let $g_\ell$ denote the number of Toffoli gates in logical layer $\ell$. Since at most $P_M(\mathcal{T})$ Toffoli gates can be executed in parallel in a single placement step, a layer containing $g_\ell$ many Toffoli gates must be partitioned into $\left\lceil g_\ell/ P_M(\mathcal{T}) \right\rceil$ batches. Each batch incurs a routing cost of at most $D_\pi(\mathcal{T})$. Therefore, we obtain
	\begin{equation}
		D_M(\mathcal{T}) \le D_{\pi}(\mathcal{T})\cdot
		\sum_{\ell=1}^{D_{\mathrm{logical}}}
		\left\lceil \frac{g_\ell}{P_M(\mathcal{T})} \right\rceil.
		\label{eqn:gen_form}
	\end{equation}
	Using the inequality $\lceil x \rceil < x+1$, we obtain
	\begin{align*}
		\sum_{\ell=1}^{D_{\mathrm{logical}}}
		\left\lceil
		\frac{g_\ell}{P_M(\mathcal{T})}
		\right\rceil
		&<
		\sum_{\ell=1}^{D_{\mathrm{logical}}}
		\left(
		\frac{g_\ell}{P_M(\mathcal{T})}+1
		\right)\\
		&=
		\left(\frac{1}
		{P_M(\mathcal{T})}\sum_{\ell=1}^{D_{\mathrm{logical}}} g_\ell\right) + D_{\mathrm{logical}}\\
		&=
		\frac{T_g(n)}{P_M(\mathcal{T})} + D_{\mathrm{logical}}.
	\end{align*}
	Substituting this in Eq. \ref{eqn:gen_form}, we obtain
	$$D_M(\mathcal{T}) < D_{\pi}(\mathcal{T})
	\left[
	\frac{T_g(n)}{P_M(\mathcal{T})} + D_{\mathrm{logical}}
	\right].$$
	Furthermore, if $g_\ell \le P_M(\mathcal{T})$ for every logical layer $\ell$, then each layer can be executed in a single batch, and (\ref{eqn:gen_form}) simplifies to
	$D_M(\mathcal{T})
	\le
	D_{\pi}(\mathcal{T})\cdot D_{\mathrm{logical}}.$
\end{proof}

\begin{figure*}[htbp]
	\centering
	\begin{subfigure}{0.32\textwidth}
		\centering
		\includegraphics[scale=1]{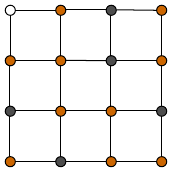}
		\caption{$P_3$ packing on 2D square-grid topology.}
		\label{fig:p3_grid_packing}
	\end{subfigure}
	\hfill
	\begin{subfigure}{0.32\textwidth}
		\centering
		\includegraphics[scale=1]{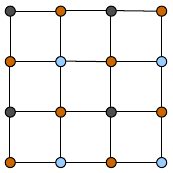}
		\caption{$C_4$ packing on 2D square-grid topology.}
		\label{fig:c4_grid_packing}
	\end{subfigure}
	\hfill
	\begin{subfigure}{0.32\textwidth}
		\centering
		\includegraphics[scale=1]{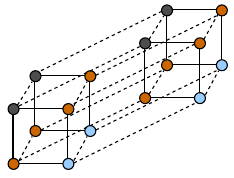}
		\caption{$C_4$ packing on Hypercube topology.}
		\label{fig:c4_hypercube_packing}
	\end{subfigure}
	\caption{Examples of $P_3$ and $C_4$ packings on different topologies.}
	\label{fig:packing_comparison}
\end{figure*}

\subsection{Depth overhead bound for logical-AND based decomposition}
\label{subsec:gidney_packing}
The decomposition shown in FIG.~\ref{fig:logicalAND-2d} corresponds to a $P_3$ subgraph, one of the most common motifs encountered in practical qubit topologies.

\begin{proposition}
	\label{prop:prop1}
	Let $\mathcal{G}=(V, E)$ be a graph that admits a Hamiltonian path. Then the maximum packing number of $P_3$ in $\mathcal{G}$ satisfies $P(\mathcal{G})\ge \lfloor |V|/3\rfloor$. Moreover, this bound is tight. Consequently, the maximum number of vertex-disjoint copies of $P_3$ in $\mathcal{G}$ is $\lfloor |V|/3\rfloor$.
\end{proposition}
\begin{proof}
	Let $|V|=k$. Since each vertex-disjoint copy of $P_3$ occupies three vertices, any packing can contain at most $\lfloor k/3\rfloor$ copies.
	Let $v_1,v_2,\ldots,v_k$ be a Hamiltonian path in $\mathcal{G}$. Partition the path into consecutive blocks of three vertices:
	$(v_1, v_2, v_3),\ (v_4, v_5, v_6)$, and so on. Each block induces a copy of $P_3$, and all such copies are vertex-disjoint. Therefore, $\mathcal{G}$ contains at least $\lfloor k/3\rfloor$ vertex-disjoint copies of $P_3$. Combining this with the upper bound proves that the maximum number of vertex-disjoint copies of $P_3$ in $\mathcal{G}$ is exactly $\lfloor |V|/3 \rfloor$.
\end{proof}

Although motivated by square-grid architectures, Proposition~\ref{prop:prop1} applies to any graph containing a Hamiltonian path. In such cases, the $P_3$ motifs can be placed directly along the Hamiltonian path. This observation yields the following result.
\begin{theorem}
	\label{lem:lemma2}
	Consider the mapping of an optimal depth $n$-MCT decomposition from \cite{dutta2025pra} onto a $p\times q$ dimensional square grid topology using the $P_3$ motif. Then the mapping depth satisfies
	$$D_{P_3}(\mathcal{T}) < D_{\pi}(\mathcal{T})
	\left[
	\frac{3}{2} + \frac{3}{4n-6} + \left\lceil \log_2 n \right\rceil
	\right].$$
\end{theorem}
\begin{proof}
	From Lemma \ref{lem:lemma1} and Proposition \ref{prop:prop1}, we obtain
	$$D_{P_3}(\mathcal{T}) < D_{\pi}(\mathcal{T})
	\left[
	\frac{n-1}{\left\lfloor \frac{pq}{3}\right\rfloor} + \left\lceil\log_2 n\right\rceil
	\right].$$
	From \cite{dutta2025pra}, an optimal depth $n$-MCT decomposition requires $2n-1$ qubits. Hence, we assume that $pq\geq 2n-1$, which implies
	$\lfloor pq/3\rfloor \ge \lfloor 2n-1/3\rfloor$.
	Therefore,
	\begin{equation}
		D_{P_3}(\mathcal{T}) < D_{\pi}(\mathcal{T})
		\left[
		\frac{n-1}{\left\lfloor\frac{2n-1}{3}\right\rfloor} + \left\lceil\log_2 n\right\rceil
		\right].
		\label{eq:Dp3}
	\end{equation}
	Define $A(n)=
	\frac{n-1}{\left\lfloor\frac{2n-1}{3}\right\rfloor}$.
	To bound $A(n)$, we consider the following cases $n=3m$, $3m+1$, and $3m+2$, where $m\in\mathbb{N}$.\\
	
	\noindent Case I: $n=3m$. Then,
	\begin{align*}
		A(n) = \frac{3m-1}{\lfloor\frac{6m-1}{3}\rfloor} = \frac{3m-1}{2m -1}
		= & \frac{3}{2} + \frac{1}{2(2m-1)}\\
		= & \frac{3}{2} +\frac{3}{4n-6}.
	\end{align*}
	\noindent Case II: $n=3m+1$. Then,
	\begin{align*}
		A(n) = \frac{3m}{\lfloor\frac{6m +1}{3}\rfloor} = \frac{3m}{2m} = \frac{3}{2}.
	\end{align*}
	\noindent Case III: $n=3m+2$. Then,
	\begin{align*}
		A(n) = \frac{3m+1}{\lfloor\frac{6m +3}{3}\rfloor} = \frac{3m+1}{2m+1} = & \frac{3}{2} - \frac{1}{2(2m+1)}\\
		= &\frac{3}{2} - \frac{3}{4n-2}.
	\end{align*}
	Combining the above cases, we obtain
	$$A(n) \leq \frac{3}{2} + \frac{3}{4n-6}.$$
	Hence, substituting $A(n)$ in Eq. \ref{eq:Dp3}, we obtain
	$$D_{P_3}(\mathcal{T}) < D_\pi(T) \left[\frac{3}{2} + \frac{3}{4n-6} + \left\lceil \log_2 n \right\rceil \right].$$
	The claim is shown.
\end{proof}

\subsection{Depth overhead bound for T-depth-1 decomposition}
\label{sub:jaques_packing}
The decomposition shown in FIG.~\ref{fig:Jaques-2d} corresponds to a $C_4$ subgraph. Such a motif occurs naturally in several architectures, including 2D square grids, hypercubes, and Benes-inspired butterfly networks \cite{brierly2017qic}. However, given their practical relevance to current quantum hardware, we primarily consider 2D square grids, while also analyzing the more highly connected hypercube topology.

\subsubsection{2D square grid}
\label{subsub:depth-2dsq}
A 2D square grid is a natural architecture for hosting $C_4$ motifs. Given a 2D square grid topology $\mathcal{T}$ of dimension $p \times q$, here we derive an upper bound on the depth overhead $D_{C_4}(\mathcal{T})$ required to decompose an $n$-MCT gate using the optimal Toffoli-depth mapping strategy from \cite{dutta2025pra}. 
\begin{proposition}
	\label{prop:prop2}
	Let $\mathcal{T}$ be a $p \times q$ square-grid topology. Then the maximum number of pairwise vertex-disjoint $C_4$ motifs in $\mathcal{T}$ is given by $P_{C_4}(\mathcal{T}) = \left\lfloor p/2 \right\rfloor \cdot \left\lfloor q/2 \right\rfloor$.
\end{proposition}
\begin{proof}
	A vertex-disjoint packing consists of alternating squares in both the row and column directions, ensuring that no two selected $C_4$ motifs share a common vertex. Consequently, along the row direction, the number of such squares is at least $\lfloor p/2 \rfloor$, while along the column direction it is at least $\lfloor q/2 \rfloor$. Therefore,
	$$P_{C_4}(\mathcal{T}) \geq \left\lfloor p/2\right \rfloor \left \lfloor q/2 \right \rfloor.$$
	
	To derive the upper bound, consider the arrangement shown in FIG.~\ref{fig:c4_grid_packing}. Along the row and column directions, there are $p-1$ and $q-1$ elementary squares, respectively. Since adjacent squares share vertices, a vertex-disjoint packing can include at most every alternate square in each direction. Therefore,
	$$P_{C_4}(\mathcal{T}) \leq \left\lceil\frac{p-1}{2} \right \rceil\  \left \lceil\frac{q-1}{2} \right \rceil = \left \lfloor \frac{p}{2}\right \rfloor \left \lfloor \frac{q}{2} \right \rfloor.$$
	
	Combining this upper bound with the previously established lower bound yields
	$P_{C_4}(\mathcal{T}) = \lfloor p/2\rfloor\lfloor q/2\rfloor$.
\end{proof}

The following theorem derives a bound on the mapping depth as follows.
\begin{theorem}
	\label{lem:lemma3}
	Let $\mathcal{T}$ be a $p \times q$ square-grid topology with packing number $P_{C_4}(\mathcal{T})$. For the optimal Toffoli-depth decomposition of an $n$-MCT gate as described in \cite{dutta2025pra}, the mapping depth satisfies
	$$D_{C_4}(\mathcal{T}) < D_{\pi}(\mathcal{T})\left[\frac{16(n-1)}{2n+1} + \lceil\log_2n\rceil \right].$$
\end{theorem}
\begin{proof}
	From Lemma \ref{lem:lemma1}, substituting $D_{\text{logical}} = \lceil \log_2n \rceil$ and $T_g(n) = n-1$, we obtain
	\begin{equation}
		D_{C_4}(\mathcal{T}) < D_{\pi}(\mathcal{T})\left[\frac{n-1}{P_{C_4}(\mathcal{T})} + \lceil\log_2n\rceil \right].
		\label{eq:depth-c4}
	\end{equation}
	
	From Proposition \ref{prop:prop2}, we have $P_{C_4}(\mathcal{T}) = \left\lfloor p/2 \right\rfloor \left\lfloor q/2 \right\rfloor$. Since $p,q \ge 2$, it follows that $\left\lfloor p/2 \right\rfloor \geq p/4$ and $\left\lfloor q/2 \right\rfloor \geq q/4$. Therefore, 
	$P_{C_4}(\mathcal{T}) \geq pq/16$. Substituting this in Eq.~\ref{eq:depth-c4}, we obtain
	$$D_{C_4}(\mathcal{T}) < D_{\pi}(\mathcal{T})\left[\frac{16(n-1)}{pq} + \lceil\log_2n\rceil \right].$$
	Additionally, since in each layer at most $\lfloor n/2 \rfloor$ Toffoli gates can be executed in parallel, and following FIG.~\ref{fig:Jaques-2d}, each Toffoli decomposition requires one reusable ancilla (which is freed before the next layer), the total number of ancilla required is at most $2\lfloor n/2 \rfloor \le n$. Assuming further that the topology is large enough to accommodate both computation and ancilla qubits, i.e.,
	$pq \geq  (n+1)+n = 2n+1$, 
	we obtain
	$$D_{C_4}(\mathcal{T}) < D_{\pi}(\mathcal{T})\left[\frac{16(n-1)}{2n+1} + \lceil\log_2n\rceil \right].$$
	This completes the proof.
\end{proof}
In the following section, we derive the similar bounds on hypercubes.

\subsubsection{Hypercubes} 
\label{subsub:hypercube}
The $r$-dimensional hypercube, denoted by $Q_r = (V, E)$, is a graph whose vertex set consists of all possible binary strings of length $r$, i.e., $V_{Q_r}=\{(a_1,a_2,\ldots,a_r)\mid a_i\in\{0,1\}\}$. Two vertices are adjacent if and only if they differ in exactly one coordinate. That is, for two vertices $u=(u_1,\ldots,u_r)$ and $v=(v_1,\ldots,v_r)$, $\{u,v\}\in E_{Q_r}$ if and only if there exists a unique index $i\in\{1,\ldots,r\}$ such that $u_i \neq v_i$ and $u_j = v_j$ for all $j\neq i$. In this regard, we have the following result in terms of packing number.
\begin{proposition}
	\label{prop:prop3}
	The maximum number of vertex-disjoint 4-cycles ($C_4$) in an $r$-dimensional hypercube $Q_r$ is $2^{r-2}$.
\end{proposition}
\begin{proof}
	Without considering structural constraints, the maximum value of $P_{C_4}(Q_r)$ is $|V|/4 = 2^{r-2}$, i.e., $P_{C_4}(Q_r) \leq 2^{r-2}$. To prove that this bound is tight, fix two coordinate directions, say $i=1$ and $i=2$. Each 4-cycle (square) can then be represented as
	$(0,0,\mathbf{z}),\ (1,0,\mathbf{z}),\ (0,1,\mathbf{z}),\ (1,1,\mathbf{z})$, where $\mathbf{z} = (a_3, a_4, \dots, a_r)$. These squares are vertex-disjoint whenever the choices of $\mathbf{z}$ are distinct. Hence, the number of disjoint squares is determined by the number of possible $\mathbf{z}$ assignments, which is $2^{r-2}$. Therefore, $P_{C_4}(Q_r) = 2^{r-2}$.
\end{proof}

In this direction, we have the following result.
\begin{theorem}
	\label{lem:lemma4}
	Let $Q_r$ be an $r$-dimensional hypercube. For an optimal depth $n$-MCT decomposition, the mapping depth satisfies
	$$D_{C_4}(Q_r) < D_{\pi}(Q_r) \lceil \log_2n \rceil.$$
\end{theorem}
\begin{proof}
	Since the maximum number of Toffoli gates that can be executed in parallel is $\lfloor n/2 \rfloor$, we have $g_{\ell}\leq \lfloor n/2 \rfloor$. Moreover, the decomposition in FIG.~\ref{fig:Jaques-2d} requires one additional ancilla per Toffoli gate. Therefore, executing $\lfloor n/2 \rfloor$ Toffoli gates in parallel requires an additional $\lfloor n/2 \rfloor$ ancilla qubits, resulting in a total requirement of $2n+1$ qubits. Assuming that the topology is sufficiently large to accommodate the entire circuit, i.e., $2^r \geq 2n+1$, it follows that $2^{r-2} \geq n/2 + 1/4$. Hence, $P_{C_4}(Q_r) > n/2$, and the corresponding bound becomes trivial:
	$$\left\lceil \frac{g_\ell}{P_{C_4}(Q_r)} \right\rceil = 1.$$
	Putting $D_{\text{logical}} = \lceil \log_2 n\rceil$, the mapping depth becomes
	$$D_{C_4}(Q_r) < D_{\pi}(Q_r) \sum_{\ell=1}^{D_{\text{logical}}} \left\lceil \frac{g_\ell}{P_{C_4}(Q_r)} \right\rceil = D_{\pi}(Q_r) \lceil \log_2 n\rceil.$$
\end{proof}

\subsection{Routing complexity}
\label{sub:complexity}
The routing bound has been extensively studied in several works, as noted in Section \ref{sec:intro}. It is commonly expressed in terms of the routing number of a graph $G$, denoted by $rt(G)$ \cite{childs2019tqc, alon1994siam}. In this section, we combine routing bounds $D_{\pi}(\mathcal{T})$ with packing bounds to obtain the total bound $D_M(\mathcal{T})$. In particular, we consider 2D square grids of size $(p \times q)$, which have routing complexity $\mathcal{O}(2p + q)$ \cite{alon1994siam}, and hypercubes $Q_r$ of dimension $r$, for which bitonic sorting networks yield routing complexity $r(r+1)/2 = \mathcal{O}(\log^2 n)$ \cite{beals2013rspa}.

\section{Experimental Evaluation of Motif-Based Placement}
\label{sec:experiments}
So far, we have derived depth overhead bounds for mapping optimal-depth $n$-MCT decompositions onto grids and hypercubes for which optimal packing bounds can be established. In this section we validate those bounds using a routing-aware placement strategy that incorporates optimal packing, and evaluate their effectiveness againsta the upper-bound. 

We also consider cases for which optimal packing is not provable. Various packing heuristics are compared across different topologies and assess their performance relative to optimal solutions obtained via exhaustive search.

\subsection{Experiments for optimal-packing bounds}
\label{sub:optimal-packing}
As shown in Section~\ref{subsec:gidney_packing}, we derive the optimal packing bounds for $P_3$ and $C_4$ motifs on 2D square grids and hypercubes. To establish these bounds, we employ a set of heuristics that facilitate efficient placement within the optimally selected motifs.

For each batch of Toffoli gates executable within a single layer, logical Toffoli gates are assigned to available physical motif slots using a greedy, routing-aware placement strategy. At each iteration, the algorithm evaluates all unplaced gates, unoccupied motif slots, and valid role assignments within each slot, and selects the assignment with the minimum heuristic cost. The selected gate and motif slot are then removed from further consideration, and the process is repeated until all gates in the batch have been placed. The three cost functions used in our experiments are summarized in TABLE~\ref{tab:placement_heuristics}.
\begin{table*}[htbp]
	\centering
	\renewcommand{\arraystretch}{1.35}
	\begin{tabular}{c c c c}
		\hline
		\hline
		~Heuristic~ & ~~Cost Function \(C_H(\phi)\)~~ & ~Routing Proxy~ & ~Interpretation~ \\[0.1cm]
		\hline
		\(H_1\) &
		\(\displaystyle \sum_{q\in Q(B)} d_T(\pi(q),\phi(q))\) &
		Total displacement &
		~Keeps the batch close to the current placement.~\\[0.1cm]
		\(H_2\) &
		\(\displaystyle \max_{q\in Q(B)} d_T(\pi(q),\phi(q))\) &
		~Worst-case displacement~ &
		Limits the farthest-moving qubit.\\[0.1cm]
		\(H_3\) &
		\(\displaystyle \alpha C_{H_1}(\phi)+\beta C_{H_2}(\phi)\) &
		Balanced displacement &
		Combines total and worst-case movement. \\[0.1cm]
		\hline
	\end{tabular}
	\caption{Routing-aware placement heuristics used for selecting motif placements for each batch.}
	\label{tab:placement_heuristics}
\end{table*}

These heuristics use shortest-path displacement as a lightweight surrogate for routing complexity, avoiding the computational overhead of explicitly estimating the routed depth for every candidate placement.

\subsubsection{Verification of the formulated Bounds}
\label{subsub:verification}
The analytical bounds derived in Section~\ref{sec:depth_bound} are validated through experiments employing routing-aware placement heuristics, as shown in FIG.~\ref{fig:packing_bound_comparison}. In each case, the observed depth overhead is compared against the corresponding theoretical upper bound. For every motif, two plots are presented, corresponding to the minimum and maximum depth bounds obtained after mapping the motif onto a grid topology. The maximum-bound plots consider all grids that satisfy the required size constraints and therefore capture the maximum possible depth overhead. In contrast, the minimum-bound plots illustrate the smallest achievable overhead under the same constraints.
\begin{figure*}[htbp]
	\centering
	\begin{subfigure}{0.48\linewidth}
		\centering
		\includegraphics[width=\linewidth]{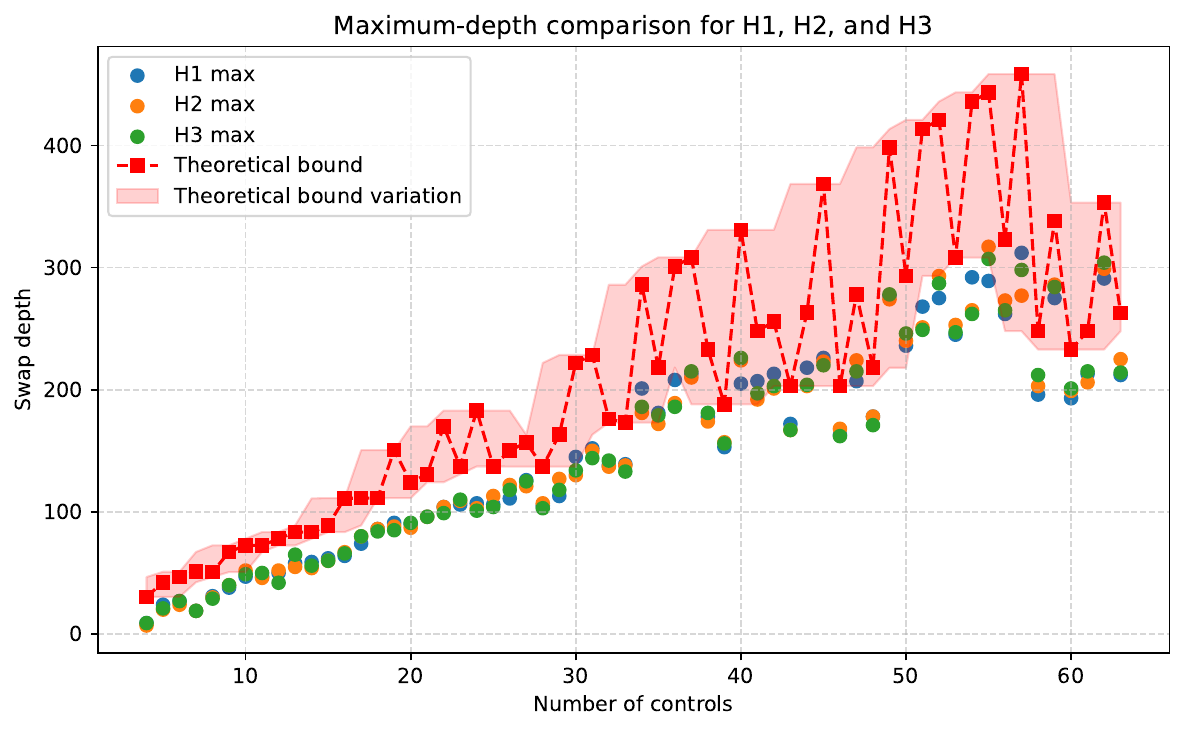}
		\caption{$P_3$ packing bounds for layouts giving maximum depth overhead.}
		\label{fig:bound_a}
	\end{subfigure}
	\hfill
	\begin{subfigure}{0.48\linewidth}
		\centering
		\includegraphics[width=\linewidth]{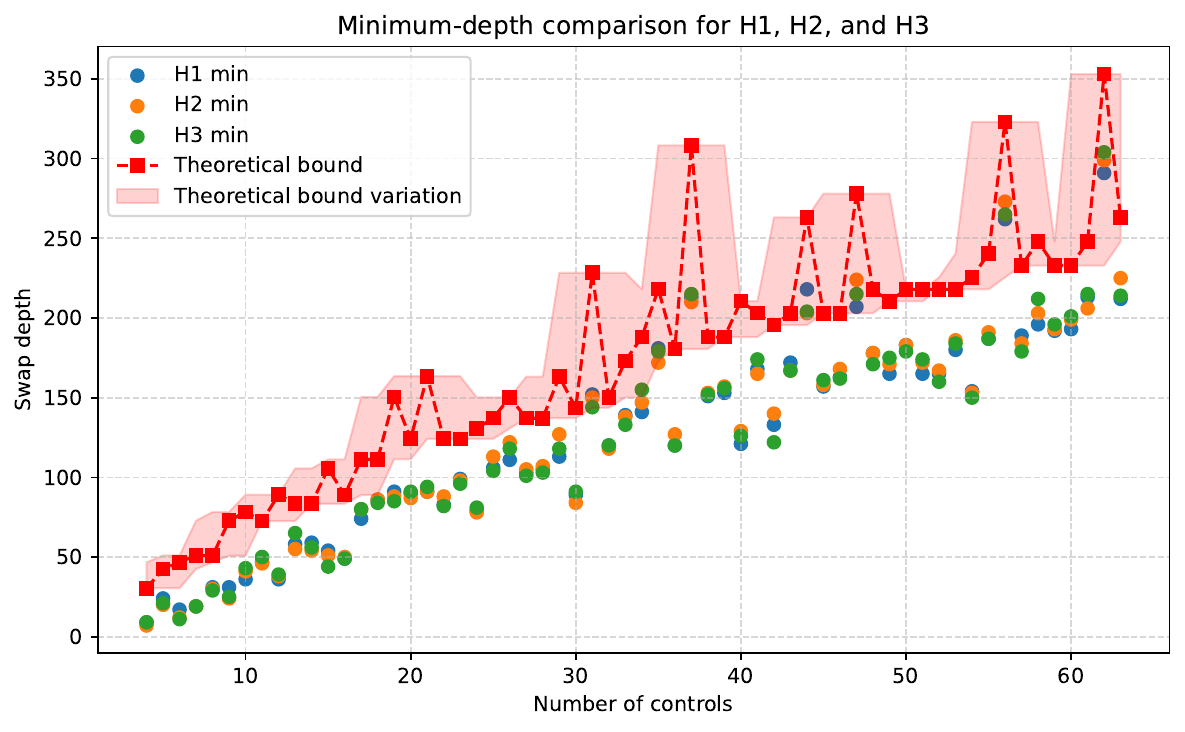}
		\caption{$P_3$ packing bounds for layouts giving  minimum depth overhead.}
		\label{fig:bound_b}
	\end{subfigure}
	\vspace{0.35cm}
	\begin{subfigure}{0.48\linewidth}
		\centering
		\includegraphics[width=\linewidth]{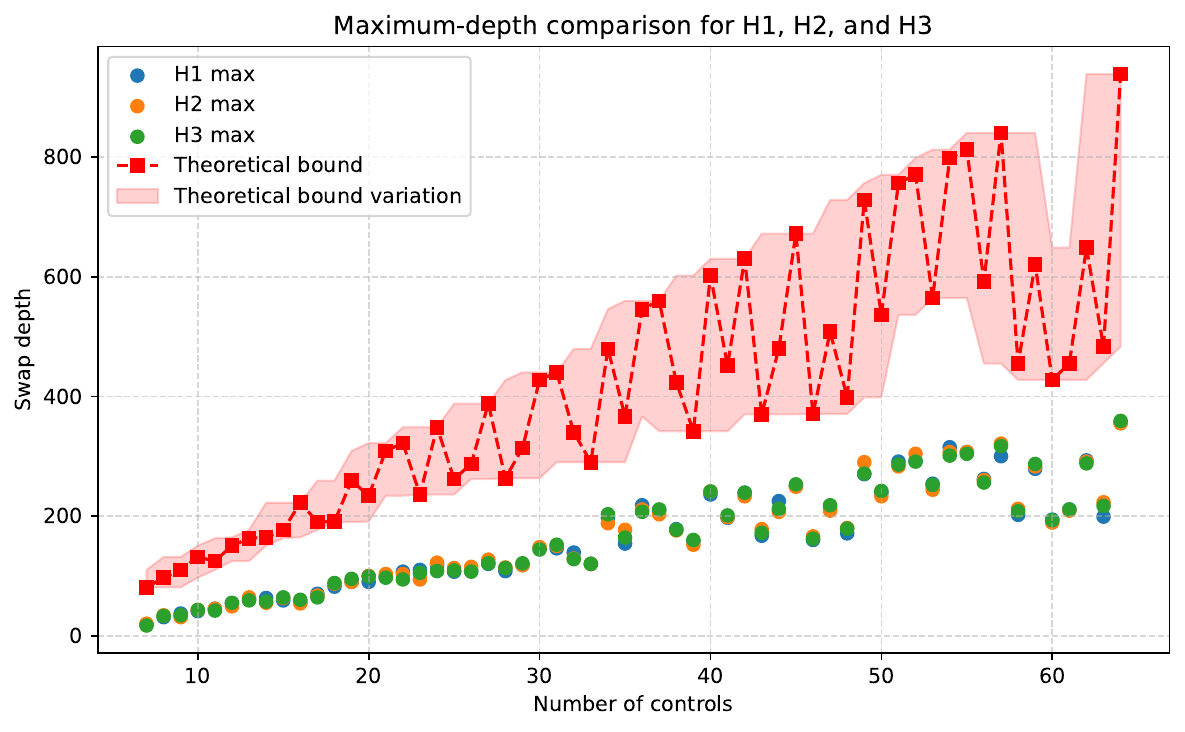}
		\caption{$C_4$ packing bounds for layouts giving  maximum depth overhead.}
		\label{fig:bound_c}
	\end{subfigure}
	\hfill
	\begin{subfigure}{0.48\linewidth}
		\centering
		\includegraphics[width=\linewidth]{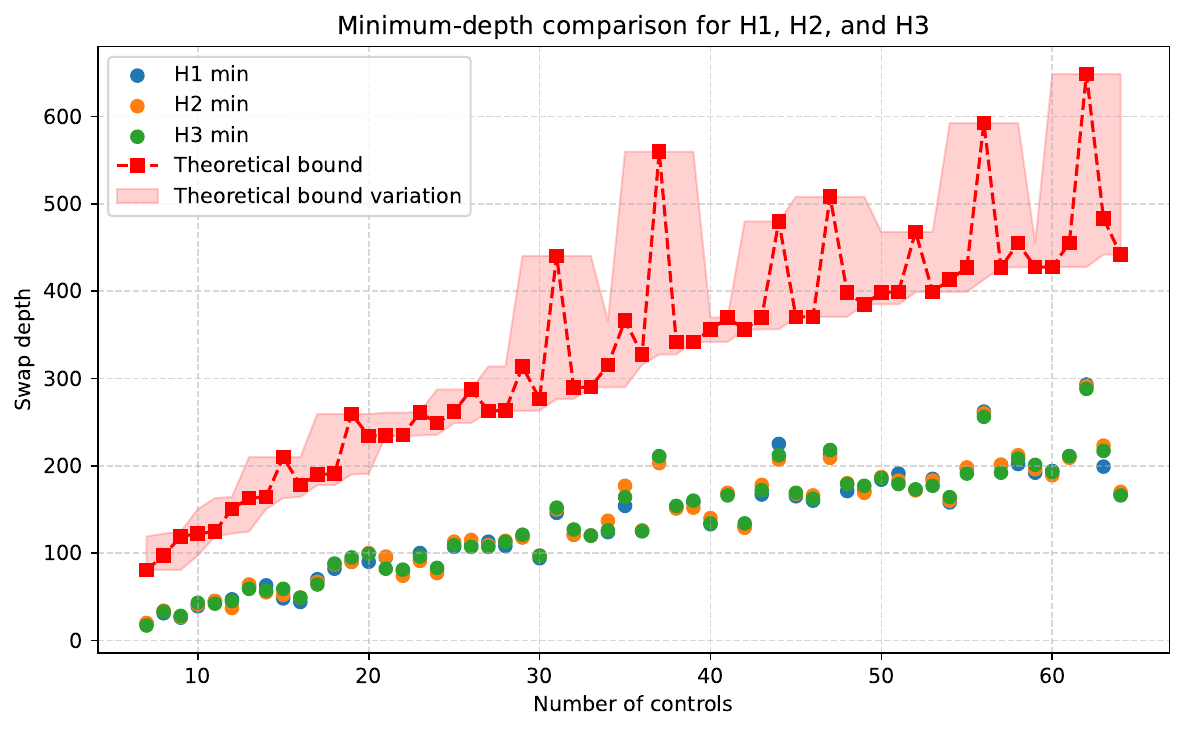}
		\caption{$C_4$ packing bounds for layouts giving  minimum depth overhead.}
		\label{fig:bound_d}
	\end{subfigure}
	\caption{Comparison of the depth overhead bounds obtained using $P_3$ and $C_4$ packing strategies on 2D square-grid layouts. For each packing motif, the bounds are reported for both the maximum-depth and minimum-depth topology among all feasible layouts. The subfigures (a) and (b) show the maximum and minimum depth overhead bounds for $P_3$ packing, respectively, while subfigures (c) and (d) show the corresponding bounds for $C_4$ packing.}
	\label{fig:packing_bound_comparison}
\end{figure*}

In Figure \ref{fig:bound_hyp} We also see similar bounds for $C_4$ packing in hypercube with minimum possible size ($r$), the placement heuristics try to minimize routing costs and we can see that the bound is satisfied in this case considering the overall bound as mentioned in Section \ref{sec:depth_bound}. The step nature of the function and the results are because of the $\lceil \log_2 n\rceil$ in the bound as mentioned in Theorem \ref{lem:lemma4}.

\begin{figure}
	\centering
	\includegraphics[width=\linewidth]{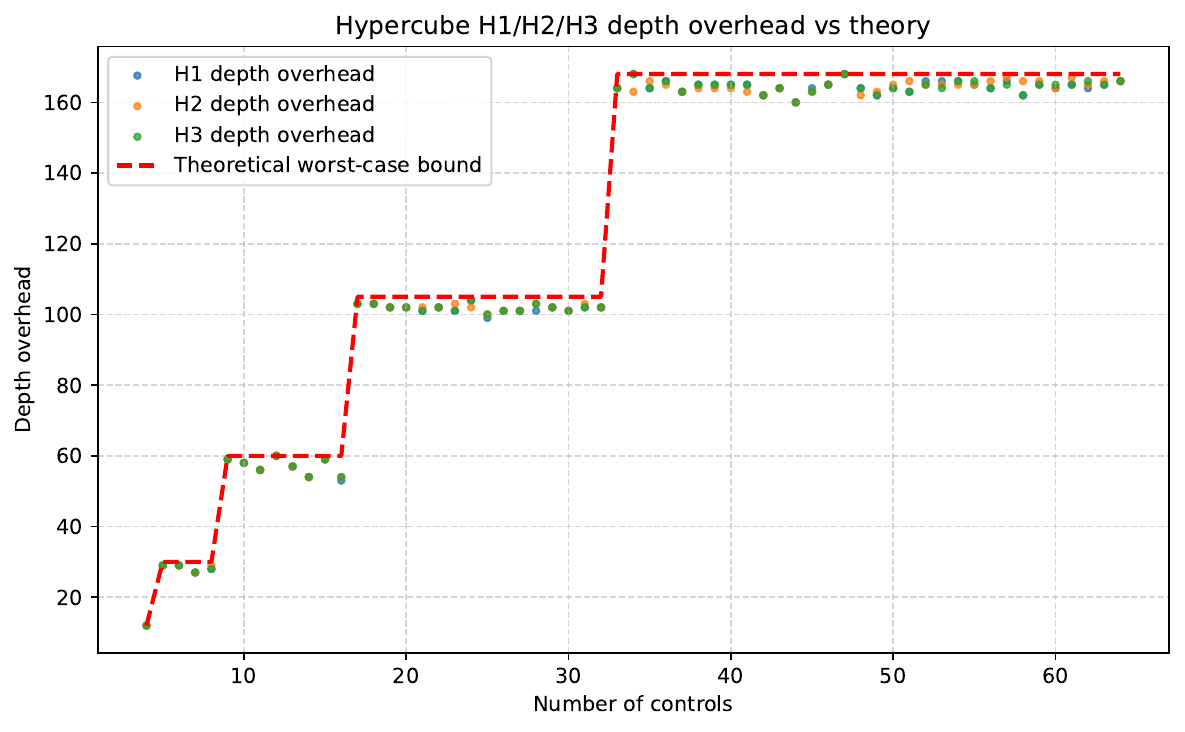}
	\caption{Depth overhead bound for packing $C_4$ motifs in hypercube layouts.}
	\label{fig:bound_hyp}
\end{figure}
Across all evaluated instances, the measured depth overhead remains below the theoretical upper bound. Although the bound is guaranteed to hold for all valid placement and routing configurations, the routing-aware heuristics frequently exploit favorable placement opportunities, resulting in significantly lower overheads in practice.

These experimental results demonstrate that the derived upper bounds are conservative and robust. In particular, the observed performance of the placement and routing heuristics consistently satisfies the proposed bounds, thereby validating their correctness and practical applicability.

\subsubsection{Comparison with other layout mechanisms}
\label{subsub:comparison}
We compare IBM's SABRE Layout and Dense Layout with the proposed packing-based placement methods, using the routing framework of \cite{alon1994siam}. To isolate the effect of placement from that of routing, all placement strategies are evaluated using the same permutation-routing backend. This choice is motivated by the fact that placement and routing are inherently coupled optimization problems, and routing heuristics such as SABRE are typically designed to operate in conjunction with their corresponding placement algorithms. Consequently, combining a packing-based placement strategy with an unrelated routing heuristic could obscure the true impact of the placement method.

For consistency, we therefore employ a common routing backend for all placement approaches and evaluate the resulting depth overhead using the worst-case permutation-routing bound for two-dimensional grids, namely $\mathcal{O}(2p+q)$ for a $p \times q$ grid with $p \leq q$. Under this framework, the observed differences can be attributed primarily to the quality of the generated layouts rather than to the routing methodology. The corresponding experimental results are presented in FIG.~\ref{fig:sabre_dense_comparison}.

\begin{figure*}[htbp]
	\centering
	
	\begin{subfigure}{0.48\linewidth}
		\centering
		\includegraphics[width=\linewidth]{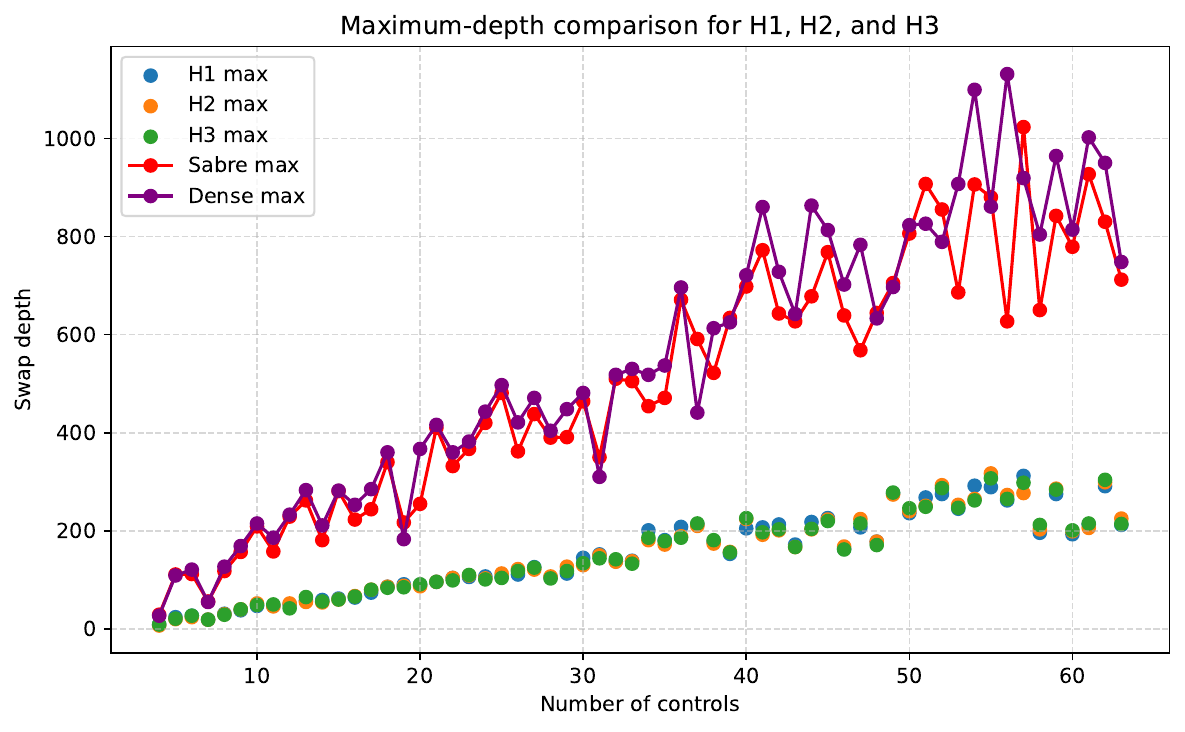}
		
		\caption{Maximum depth $P_3$ packing comparison with SABRE/ Dense layout.}
		\label{fig:bound_e}
	\end{subfigure}
	\hfill
	\begin{subfigure}{0.48\linewidth}
		\centering
		\includegraphics[width=\linewidth]{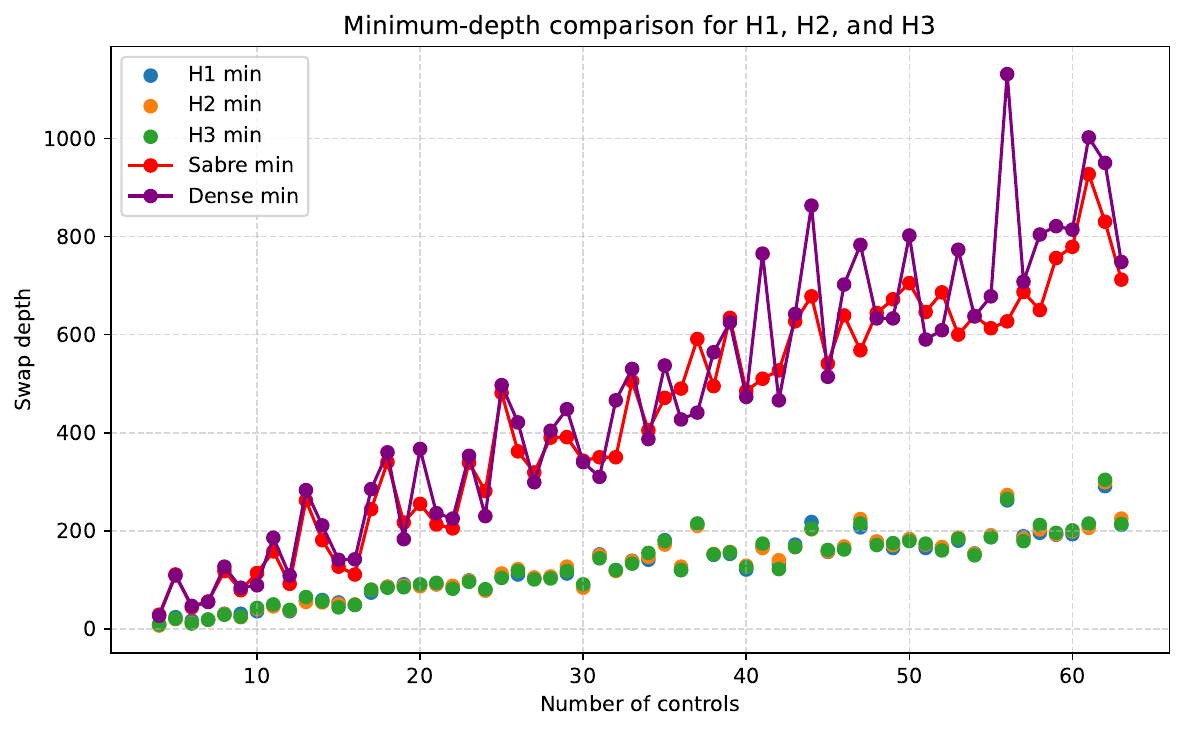}
		
		\caption{Minimum depth $P_3$ packing comparison with SABRE/ Dense layout.}
		\label{fig:bound_f}
	\end{subfigure}
	
	\vspace{0.35cm}
	
	\begin{subfigure}{0.48\linewidth}
		\centering
		\includegraphics[width=\linewidth]{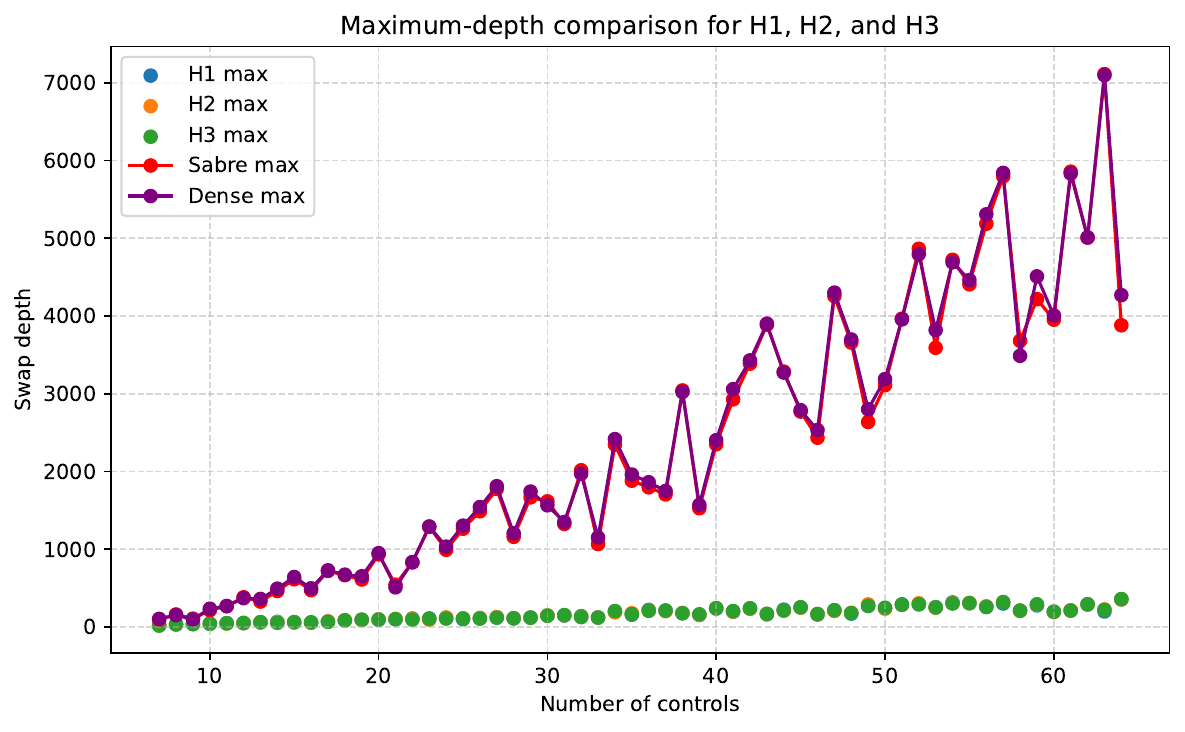}       
		\caption{Maximum depth $C_4$ packing comparison with SABRE/ Dense layout.}
		\label{fig:bound_g}
	\end{subfigure}
	\hfill
	\begin{subfigure}{0.48\linewidth}
		\centering
		\includegraphics[width=\linewidth]{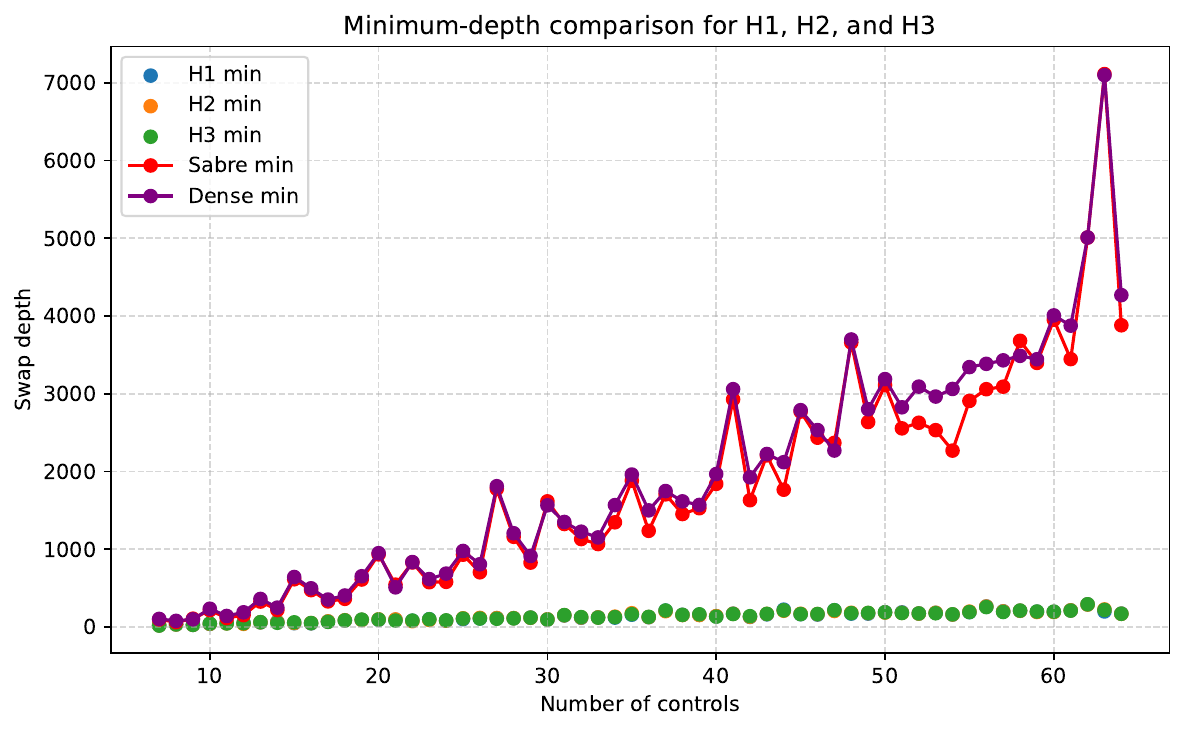}
		\caption{Minimum depth $C_4$ packing comparison with SABRE/ Dense layout.}
		\label{fig:bound_h}
	\end{subfigure}
	\caption{Comparison of the maximum and minimum depth overheads achieved by the proposed packing-based placement strategies and those obtained using the SABRE and Dense layouts from IBM. Subfigures (a) and (b) show comparisons for the $P_3$ packing strategy, while subfigures (c) and (d) present results for the $C_4$ packing strategy, using the routing upper-bound framework of Alon et al.~\cite{alon1994siam}.}
	\label{fig:sabre_dense_comparison}
\end{figure*}
\subsection{Heuristic packing for topologies without closed-form bounds}
\label{sub:heuristic packing}
Packing motifs into arbitrary graphs does not, in general, admit a closed-form bound. In our experiments, we select motifs based on interaction graphs arising from Toffoli-to-T decompositions \cite{amy2013ieee,gidney2018quantum}, as well as geometric structures in specific topologies and Toffoli constructions. In particular, we observe such a structure in the 4-MCT construction of \cite{leiserson1980htree} (FIG.~\ref{fig:h4}) within hexagonal and heavy-hex architectures, and similarly in the X-type stabilizers for the hexagonal code \cite[FIG.~2]{chamberland2020prx}. These structures enable higher-order MCT decompositions with reduced depth overhead through effective grouping.

When no formal packing bound is available, we employ branch-and-bound techniques to explore the space of feasible configurations and identify optimal packings. However, this approach is computationally expensive in both memory and runtime. To address this, we adopt a heuristic formulation based on a conflict graph, reducing the problem to finding a maximum independent set (MIS), as described in the following subsections. We compare the heuristic and the branch and bound capability in TABLE~\ref{tab:mis-heuristic-comparison}.

\subsubsection{Graph formulation and heuristics}
\label{subsub:graph-heuristic}
Given a graph $\mathcal{G}(V, E)$, an independent set $S$ is a subset of vertices $S \subseteq V$ such that no two vertices in $S$ are adjacent. The maximum independent set (MIS) problem seeks an independent set of maximum cardinality.

The MIS problem is a classical graph-theoretic formulation for selecting mutually non-interacting candidate sets. In our setting, this arises naturally since each gate is mapped to a subgraph location within the architecture. Each vertex of $G$ represents a potential placement for a Toffoli gate, and an edge is introduced between two vertices if their corresponding subgraphs overlap (i.e., share at least one physical qubit). Solving the MIS problem thus yields the maximum number of vertex-disjoint subgraphs in which Toffoli gates can be placed simultaneously. This problem is NP-hard and is typically addressed using exact methods, such as branch-and-bound, or heuristic approaches, including greedy and local search techniques. The heuristics used in this work are described below.

\paragraph{Min-degree MIS.}
This heuristic constructs an independent set by iteratively selecting vertices with the minimum degree (i.e., the fewest conflicts). At each step, the selected vertex is added to the solution, and all its neighbors are removed from further consideration. The intuition is that low-degree vertices are more likely to be included in a large independent set than high-degree vertices.

\paragraph{Max-degree removal MIS.}
This heuristic adopts the complementary strategy of iteratively removing the highest-degree vertices (i.e., those involved in the most conflicts). The process continues until all edges are eliminated, and the remaining vertices form an independent set corresponding to a valid packing.

\paragraph{Min-degree MIS + local search.}
Empirically, the min-degree heuristic performs well; hence, we augment it with a local improvement phase. In this refinement step, we attempt to replace a selected vertex with up to two non-selected vertices, provided they are mutually non-conflicting. The process broadens the search space and improves solution quality beyond the greedy baseline.
\begin{table*}[htbp]
	\centering
	\scalebox{0.9}{\begin{tabular}{l c c c c c}
			\hline
			\hline
			Heuristic & ~Match Rate $\uparrow$ & ~Avg.\ Ratio $\uparrow$~ & Max Gap $\downarrow$  & Topology & Packing Motif\\[0.1cm]
			\hline
			Min-degree MIS & 0.88 & 0.99 & 1 & IBM Tokyo Q-20; FIG.~\ref{fig:tokyo} & $C_3$\\
			Max-degree removal MIS & 0.53 & 0.93 & 3 & IBM Tokyo Q-20; FIG.~\ref{fig:tokyo} & $C_3$\\
			Min degree MIS + local search & 0.8846 &  0.99 & 1 & IBM Tokyo Q-20; FIG.~\ref{fig:tokyo} & $C_3$\\[1pt]
			\hline
			Min-degree MIS & 1 & 1  & 0 & Slashed-Square 2D grid; FIG.~\ref{fig:slashed} & $C_3$\\
			Max-degree removal MIS & 0.73 & 0.95 & 3 & Slashed-Square 2D grid; FIG.~\ref{fig:slashed} & $C_3$\\
			Min degree MIS + local search & 1 &  1 & 0 & Slashed-Square 2D grid; FIG.~\ref{fig:slashed} & $C_3$\\[1pt]
			\hline
			Min-degree MIS & 0.84 & 0.99 & 2 & IBM-Heavy Hex; FIG.~\ref{fig:heavy-hex} & $P_3$\\
			Max-degree removal MIS & 0.14 & 0.92 &  4 & IBM-Heavy Hex; FIG.~\ref{fig:heavy-hex} & $P_3$\\
			Min degree MIS + local search & 0.84 & 0.99 & 2 & IBM-Heavy Hex; FIG.~\ref{fig:heavy-hex} & $P_3$\\[1pt]
			\hline
			Min-degree MIS & 1 & 1 & 0 & Hexagonal; FIG.~\ref{fig:hex} & $P_3$\\
			Max-degree removal MIS & 0.5 & 0.93 &  2 & Hexagonal; FIG.~\ref{fig:hex} & $P_3$\\
			Min degree MIS + local search & 1 & 1 & 0 & Hexagonal; FIG.~\ref{fig:hex} & $P_3$\\[1pt]
			\hline
			Min-degree MIS & 0.95 & 0.992 & 1& IBM-Heavy Hex; FIG.~\ref{fig:heavy-hex} & 4-MCT H-Tree \ref{fig:h4}\\
			Max-degree removal MIS & 0.803 & 0.960 & 2 & IBM-Heavy Hex; FIG.~\ref{fig:heavy-hex} & 4-MCT H-Tree \ref{fig:h4}\\
			Min degree MIS + local search & 0.96 &  0.995 & 1 & IBM-Heavy Hex; FIG.~\ref{fig:heavy-hex} & 4-MCT H-Tree \ref{fig:h4}\\[1pt]
			\hline
			Min-degree MIS & 0.652 & 0.958 & 2 & Hexagonal; FIG.~\ref{fig:hex} & 4-MCT H-Tree FIG.~\ref{fig:h4}\\
			Max-degree removal MIS & 0.347 & 0.871 & 4 & Hexagonal; FIG.~\ref{fig:hex} & 4-MCT H-Tree FIG.~\ref{fig:h4}\\
			Min degree MIS + local search & 0.652 & 0.958 & 2 & Hexagonal; FIG.~\ref{fig:hex} & 4-MCT H-Tree FIG.~\ref{fig:h4}\\[1pt]
			\hline
	\end{tabular}}
	\caption{Comparison of MIS-based triangular-packing heuristics against the exact branch-and-bound optimum for topologies with the minimum size constraint for $n$-MCT decompositions, for $n$ ranging from 4 to 60, following \cite{dutta2025pra}.}
	\label{tab:mis-heuristic-comparison}
\end{table*}

\begin{figure*}[htbp]
	\centering
	\begin{subfigure}{0.22\textwidth}
		\centering
		\includegraphics[scale=0.35]{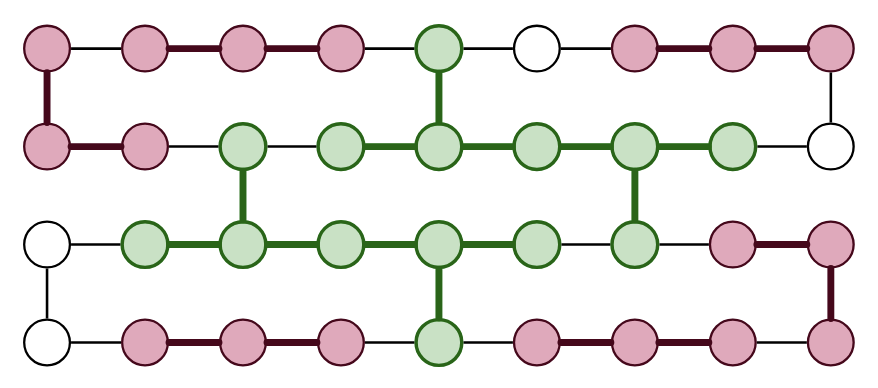}
		\caption{\textcolor{green!40!black}{4-MCT} and \textcolor{purple!40!black}{$P_3$}-Packing in IBM- heavy hex \cite{jurcevic2021qst,chamberland2020prx}}
		\label{fig:heavy-hex}
	\end{subfigure}
	\hfill
	\begin{subfigure}{0.22\textwidth}
		\centering
		\includegraphics[scale=0.35]{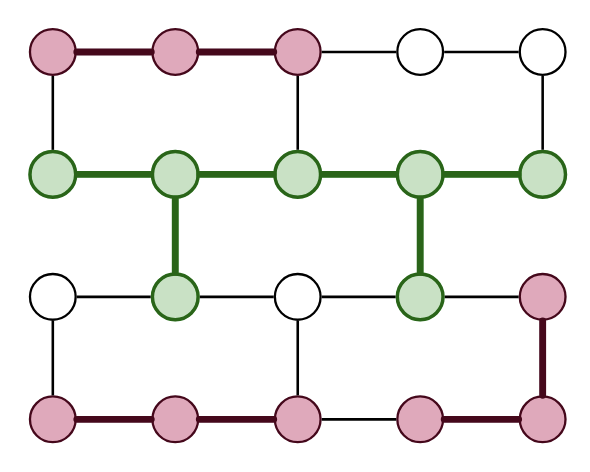}
		\caption{\textcolor{green!40!black}{4-MCT} and \textcolor{purple!40!black}{$P_3$}-Packing in hexagonal architecture}
		\label{fig:hex_packing}
	\end{subfigure}
	\begin{subfigure}{0.22\textwidth}
		\centering
		\includegraphics[scale=0.35]{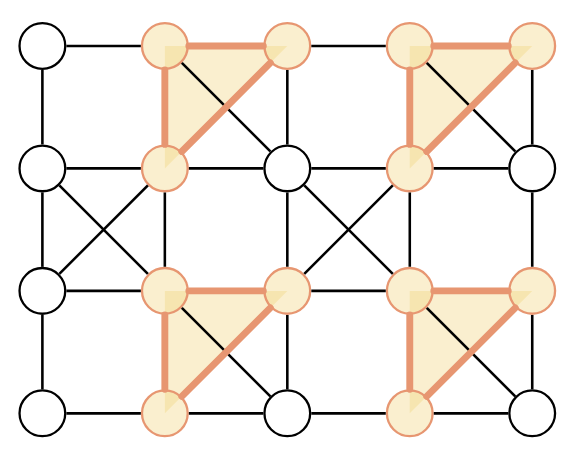}
		\caption{\textcolor{orange} {$C_3$} packing in IBM Tokyo Q-20 \cite[FIG. 2]{li2019asplos}}
		\label{fig:tokyo}
	\end{subfigure}
	\begin{subfigure}{0.22\textwidth}
		\centering
		\includegraphics[scale=0.35]{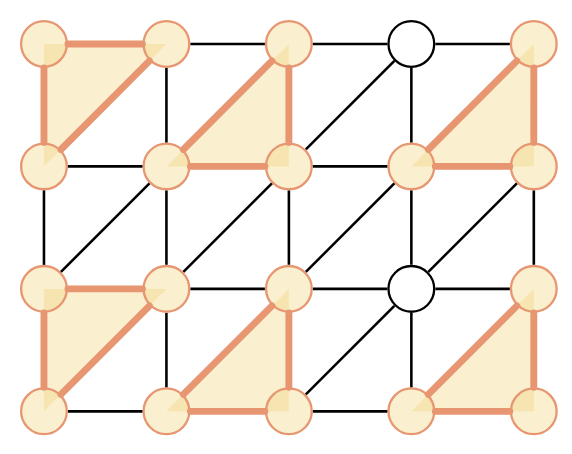}
		\caption{\textcolor{orange}{$C_3$ }packing in slashed-square 2D grid}
		\label{fig:slashed}
	\end{subfigure}
	\caption{The placement of the motif packings shown in the TABLE~\ref{tab:mis-heuristic-comparison} (last column) on various qubit topologies. FIG.~\ref{fig:heavy-hex} and FIG.~\ref{fig:hex_packing} show the placement of 4-MCT motifs (highlighted in \textcolor{green!40!black}{green}) and $P_3$ motifs (highlighted in \textcolor{purple!40!black}{purple}). Moreover, FIG.~\ref{fig:tokyo} and FIG.~\ref{fig:slashed} depict the placement of $C_3$ motifs on the Q-20 topology and the Slashed-square 2D grid, respectively, highlighted in \textcolor{orange}{yellow}.}
	\label{fig:packing_heuristic_topology}
\end{figure*}

\subsubsection{Results and inferences}
\label{subsub:results}
It can be observed in TABLE~\ref{tab:mis-heuristic-comparison} that relatively simple heuristics yield near-optimal results, indicating that the structure of the conflict graph is highly exploitable.
The metrics are match rate, average ratio, and the maximum gap. The match rate shows how often the optimal and the heuristic match, The average ratio represents average packing quality relative to optimum and the Maximum gap represents the worst deviation from optimal. 

In particular, the $P_3$, $C_3$, and H-tree motifs as demonstrated in FIG. \ref{fig:packing_heuristic_topology} exhibit favorable structural properties. Local search provides only marginal improvement over the greedy heuristic and remains largely consistent with it, further supporting the near-optimality of the greedy approach. In contrast, the maximum-removal heuristic performs significantly worse, suggesting that selecting low-conflict motifs is more effective than eliminating high-conflict ones.

The code used for all numerical simulations and experiments reported in this work is publicly available through the GitHub repository \cite{github}.

\section{Conclusion}
\label{sec:con}
In this paper, we mapped state-of-the-art optimal T-depth multi-controlled Toffoli decompositions onto 2D nearest-neighbor qubit layouts. Our analysis first considered an infinite 2D grid, where the MCT decompositions were implemented using the same number of qubits as the logical circuit through topology-aware decomposition without requiring any SWAP operations. We then turned to restricted 2D layouts with limited connectivity, where SWAP operations become unavoidable, and estimated the resulting depth overhead under different routing mechanisms.

Future work may explore mapping strategies and routing overhead bounds for different highly connected architectures like the cyclic-butterfly and other Ben\u es~\cite{benevs1964bstj} inspired networks with the packing strategy to explore for further improvement in depth complexity as told in~\cite{brierly2017qic}. Further investigation could also examine technology-aware 2D mapping approaches for general logical circuits.

\end{document}